\documentclass[11pt]{amsart}
\usepackage{geometry}                
\usepackage{graphicx}
\usepackage{amssymb, amsthm, mathrsfs}
\usepackage{epstopdf}

\usepackage{enumerate}
\usepackage{color}
\usepackage{hyperref}
\usepackage{upgreek}
\usepackage{graphicx}
\usepackage{graphics, cancel}
\usepackage{color}
\usepackage{amsmath, amsthm, slashed}
\usepackage{amssymb}
\usepackage{rotating, bm}
\usepackage{epsfig}
\usepackage{verbatim}
\usepackage{rotating}
\usepackage{graphicx}

\newcommand{\be}{\begin{equation*}}
\newcommand{\ee}{\end{equation*}}
\newcommand{\ben}[1]{\begin{equation}\label{#1}}
\newcommand{\een}{\end{equation}}
\newcommand{\bea}{\begin{eqnarray}}
\newcommand{\eea}{\end{eqnarray}}
\newcommand{\bean}{\begin{eqnarray*}}
\newcommand{\eean}{\end{eqnarray*}}

\newcommand{\R}{\mathbb{R}}

\newcommand{\rt}{\tilde{r}}

\renewcommand{\O}[1]{\mathcal{O}\left( #1 \right)}

\renewcommand{\H}{\underline{H}}

\newcommand{\df}{\tilde{\partial}}

\newcommand{\scri}{\mathscr{I}}

\newcommand{\abs}[1]{\left|#1 \right|} 
\newcommand{\norm}[2]{\left|\left |#1 \right| \right |_{#2}}

\newcommand{\tn}{\tilde{\nabla}}

\newcommand{\EKG}{Einstein--Klein-Gordon }

\newcommand{\eq}[1]{(\ref{#1})}
\newtheorem{Theorem}{Theorem}

\newtheorem{Lemma}{Lemma}
\newtheorem{Corollary}[Theorem]{Corollary}
\newtheorem{Definition}{Definition}
\newtheorem{Proposition}{Proposition}
\newtheorem{Remark}{Remark}
\numberwithin{Theorem}{section}
\numberwithin{Lemma}{section}

\title[Einstein--Klein-Gordon--AdS]{The Einstein--Klein-Gordon--AdS system for general boundary conditions}
\author{Gustav H. Holzegel}
\thanks{\texttt{g.holzegel@imperial.ac.uk} \\
\phantom{1   }\hspace{.05cm} Department of Mathematics, South Kensington Campus, Imperial College London, SW7 2AZ, UK}

\author{Claude M. Warnick}
\thanks{\vspace{.1cm} \texttt{c.warnick@warwick.ac.uk}\\
\phantom{1   }\hspace{.05cm} Department of Mathematics,  Gibbet Hill Rd, Coventry, West Midlands CV4 7AL, UK}

\date{\today \vspace{.1cm}}                                           

\begin{document}
\maketitle

\begin{abstract}
We construct unique local solutions for the spherically-symmetric Einstein--Klein-Gordon--AdS system subject to a large class of initial and boundary conditions including some considered in the context of the AdS-CFT correspondence. The proof relies on estimates developed for the linear wave equation by the second author and involves a careful renormalization of the dynamical variables, including a renormalization of the well-known Hawking mass. For some of the boundary conditions considered this system is expected to exhibit rich global dynamics, including the existence of hairy black holes. The present paper furnishes a starting point for such global investigations.
\end{abstract}

\section{Introduction}
Consider the coupled Einstein--Klein-Gordon system in the presence of a negative cosmological constant $\Lambda = - \frac{3}{l^2}$ and mass-squared $m^2= 2a$ for the Klein-Gordon field:
\begin{align}
R_{\mu \nu} - \frac{1}{2} g_{\mu \nu} R -\frac{3}{l^2} g_{\mu \nu} &= 8\pi T_{\mu \nu} \, , \nonumber \\
\Box_g \psi - \frac{2a}{l^2}\psi &=0 \, , \label{EKG}\\
 \partial_\mu \psi \partial_\nu \psi - \frac{1}{2} g_{\mu\nu} \partial_\sigma \psi \partial^\sigma \psi -\frac{a}{l^2} \psi^2 g_{\mu \nu} &= T_{\mu \nu}  \nonumber \, .
\end{align}
We wish to construct spherically symmetric solutions of (\ref{EKG}) in the class of spacetimes which are asymptotically anti de Sitter (aAdS) at infinity. The asymptotically-flat case (with $\Lambda=0$, $a=0$) has been considered in \cite{Christodoulou}  and the asymptotically-de Sitter case (with $\Lambda>0$, $a\geq 0$) in  \cite{Ringstroem, Costa} .

As is well-known, the study of hyperbolic systems  (linear or non-linear)  in aAdS spacetimes generally necessitates the prescription of boundary conditions at the timelike asymptotic infinity. In perhaps the simplest case, that of the linear wave equation 
\begin{align} \label{we}
\Box_g \psi - \frac{2a}{l^2}\psi=0
\end{align}
on a \emph{fixed} aAdS background $g$, the field has an expansion near infinity of the form:
\be
\psi \sim \psi_- \rho^{\frac{3}{2}-\kappa} + \psi_+ \rho^{\frac{3}{2}+\kappa}+\O{\rho^{\frac{5}{2}}},
\ee
where $\rho=0$ defines the conformal boundary and $\kappa = \sqrt{9/4+2a}$. For the mass-squared in the range $5/4<-2a<9/4$, the well posedness of the initial-boundary value problem with inhomogeneous Dirichlet ($\psi_-$ prescribed), Neumann ($\psi_+$ prescribed) or Robin (linear combination of $\psi_\pm$ prescribed) boundary conditions  was understood in the context of classical energy estimates in \cite{Warnick:2012fi}. For an earlier treatment of the Dirichlet case see \cite{Vasy,Holwp}. We remark that for $-2a<5/4$ there is no freedom in specifying boundary conditions, while for $-2a>9/4$, the Breitenlohner-Freedman bound \cite{BF}, the classical well-posedness theory based on energy estimates breaks down.

The connection between the linear problem (\ref{we}) and (\ref{EKG}) is that the linearization of the system (\ref{EKG}) around a \emph{fixed} spherically symmetric aAdS background $g$ yields the Klein-Gordon equation (\ref{we}).
 
 In the case of Dirichlet conditions imposed on $\psi$, in \cite{Holzegel:2011qk} the first author in collaboration with J.~Smulevici proved -- based on estimates for the linear problem \cite{Holwp} -- that the system (\ref{EKG}) was well-posed. A companion paper \cite{Holzegel:2011rk} established the stability of the Schwarzschild-AdS spacetime within this model. 
 
 With the recent results of \cite{Warnick:2012fi}, which ensure well-posedness of (\ref{we}) for general boundary conditions (and $5/4<-2a<9/4$), it is very natural to ask whether the \emph{non-linear} system (\ref{EKG}) is also well-posed for general boundary conditions. This is a non-trivial problem, because the weaker decay exhibited by $\psi$ for non-Dirichlet boundary conditions can lead to divergences in the equations for the metric coefficients.\footnote{In particular, the metrics we construct extend only at the $C^{1, \eta}$ level to the conformal boundary after rescaling for certain choices of $a$.  The problems this introduces will be resolved by a careful renormalization exploiting certain cancellations, see below.} At the level of applications, imposing these other boundary conditions will allow one to study more interesting \emph{global} dynamics such as non-trivial solitons, which are absent in the Dirichlet case.\footnote{Such boundary conditions are also of particular interest in the context of the AdS-CFT correspondence. See for instance \cite{Bhaseen:2012gg} where a black hole spacetime in a model incorporating an electromagnetic field is excited by imposing a time dependent inhomogeneous Dirichlet condition at the AdS boundary. See also \cite{Marolf:2012dr} for further examples of non-trivial boundary value problems in the AdS/CFT context.} A flavour of this is already provided by our \cite{Holzegel:2012wt}, which investigates the global dynamics of (\ref{we}) at the linear level, establishing among other things the existence of solitons for certain choices of boundary conditions. Therefore, the present paper opens the door for the mathematical analysis of the global evolution of asymptotically AdS spacetimes under physically more interesting boundary conditions, allowing, for instance, the study of stability of ``hairy black holes".

We finally mention that the global non-linear dynamics of the system (\ref{EKG}) has been intensely investigated numerically, see \cite{Bizon,Buchel:2012uh,Buchel:2013uba}.
\\ 
$\phantom{X}$ \\
{\bf Remark.} For $a=-1$, the scalar field is \emph{conformally coupled}. In this case (or more generally, for the Einstein equations coupled to any conformal matter model),  well-posedness of  the system (\ref{EKG}) can be proven without symmetry restrictions by the conformal method of Friedrich, see \cite{Friedrich, Lubbe}. However, it is not clear whether or how these methods extend to the general case.
\\ 
$\phantom{X}$ \\
{\bf New ideas and comparison with \cite{Holzegel:2011qk}.} 
In the remainder of this introduction we highlight the main difficulties and novel ideas in extending the results of \cite{Holzegel:2011qk} (homogeneous Dirichlet case) to general boundary conditions. 

We recall that a key ingredient of the argument in \cite{Holzegel:2011qk} was the consideration of a renormalized system whose well-posedness was equivalent to that of the original system. The solutions of the renormalized system were then constructed via a fixed point argument, which combined $L^2$-energy estimates for $\psi$ and (suitably weighted) pointwise estimates for the metric coefficients. Because the linear statement of \cite{Holwp} required $H^2$-regularity of solutions of the wave equation (\ref{we}), the contraction map was quite elaborate and required commutation of the wave equation, while carefully keeping track of the regularity of the metric coefficients.

The approach taken in this paper is similar (in particular the set-up of doing $L^2$-estimates for $\psi$ and pointwise estimates for the metric components 
is retained\footnote{If the boundary were at a finite distance, an approach based entirely on pointwise estimates would be possible using integration along characteristics for $\psi$. Here, unless one is in the conformally coupled case (which is essentially a ``finite" problem), it is not immediate whether and how this approach generalizes to the situation with the boundary being at infinity.}) but based on several new ingredients:
\begin{enumerate}
\item Unlike in the Dirichlet case, the energy estimates for $\psi$ have to be phrased in terms of the \emph{twisted derivatives} introduced in the linear context in \cite{Warnick:2012fi}. The twisting, while eventually enabling one to prove an energy estimate for non-Dirichlet conditions, introduces certain non-linear error-terms whose regularity and decay towards infinity has to be controlled. In addition, at several points (see for instance the formulation of the boundary condition in Section \ref{sec:bdycondition}) it becomes quite subtle whether the twisting is done with respect to a fixed boundary defining function or the inverse of the (dynamical!) geometric area radius. This difficulty is coupled with the low regularity we are working with, cf.~(4) below.
\item Unlike in the Dirichlet case, the equation for the Hawking mass \emph{also} needs to be renormalized. This may be viewed as a consequence of the fact that the usual $\partial_t$-energy for the linear problem diverges. With the renormalization one finally obtains a regular system (with a ``$\psi$-renormalized" Hawking mass (\ref{renhawk}) as a regular variable), whose contraction property can be established. 
\item Our contraction map scheme only uses the wave equation for the (inverse) area radius $\tilde{r}=\frac{1}{r}$ and the scalar field $\psi$ together with a first order equation for the renormalized Hawking mass which is integrated from the data towards the boundary. The Hawking mass can \emph{a posteriori} be shown to be constant along null-infinity in the homogeneous Dirichlet and Neumann case but 
remarkably, for Robin boundary conditions, it is in fact non-constant along the boundary, with the difference of renormalized Hawking mass between any two points on the boundary related to the (renormalized) energy flux of the scalar field  through the boundary familiar from the linear problem \cite{Warnick:2012fi}.  See Section \ref{sec:bdycondition}. We emphasize that having only three equations in the contraction map considerably simplifies the overdetermined scheme of \cite{Holzegel:2011qk}, where constancy of the Hawking mass is \emph{imposed} a-priori. 
\item Because the well-posedness statement of \cite{Warnick:2012fi} is proven at the $H^1$-level, we can close the argument with lower regularity for the contraction map than in \cite{Holzegel:2011qk}.\footnote{The lower regularity also allows us to work with the simple change of variables $\tilde{r} = 1/r$, while  \cite{Holzegel:2011qk} had to capture more detailed asymptotics.} The improvement of the regularity by commutation can be done \emph{a posteriori}. In particular, we obtain as a corollary an $H^1$-well-posedness result for the linear wave equation in a spherically symmetric background (with precise (low) regularity assumptions on the metric), see Section \ref{sec:well posed}. This may be useful for future applications. 
\item In addition, some novel estimates are obtained in the context of the contraction map, which can be directly used to simplify the proof of \cite{Holzegel:2011qk}. See Section \ref{sec:auxlemmas}.
\end{enumerate}
$\phantom{X}$ \\
{\bf Overview.} In the next section we derive the renormalized system culminating in the definition of a weak solution to the renormalized system (Definition 1). Initial and boundary data for this system are constructed in Section \ref{initdatasec} followed by a statement of the main theorem in Section \ref{sec:maintheo}. The proof of the main theorem is the content of Section 5: After definition of the relevant function spaces in Section \ref{funspacessec}, the contraction map is formulated in Section \ref{contractionsec}, with the contraction property being demonstrated in Sections  \ref{sec:auxlemmas}-\ref{sec:conmap}. In conjunction with a proposition about the propagation of the constraint equations (proven in Section \ref{sec:constraintprop}) the main theorem then follows. Generalizations of the main theorem are discussed in Section \ref{sec:generalizations} and an appropriate higher order regularity version is obtained a posteriori in Section \ref{sec:improreg}. The last section provides a useful Corollary for the linear wave equation in spherical symmetry with rough coefficients.
$\phantom{X}$ \\ \\
{\bf Acknowledgement.} The authors thank the Isaac Newton Institute for Mathematical Sciences and the organizers of the workshop ``Mathematics and Physics of the Holographic Principle" during which part of this research was carried out. We also thank Jacques Smulevici for discussions and Princeton University for its hospitality. GHH acknowledges support through NSF grant DMS-1161607.

\section{The Renormalized System of Equations}
Recall that $l^2 = -\frac{3}{\Lambda}$, where $\Lambda$ is the cosmological constant and define
\begin{equation} \label{kgmass}
\kappa = \sqrt{9/4+2a} \  \ \ \textrm{with $-9/4<2a<-5/4$} 
\end{equation}
where $2a$ is the mass-squared of the Klein-Gordon field, cf.~(\ref{EKG}).
We are interested in constructing spherically symmetric solutions of the \EKG system with a negative cosmological constant and with (possibly inhomogeneous) Dirichlet, Neumann or Robin boundary conditions. In \cite{Holzegel:2011qk} the same system was studied with homogeneous Dirichlet boundary conditions, so we may start from the following result of that paper:

\begin{Lemma}
Let $(\mathcal{M}, g, \psi)$, with $(\mathcal{M}, g)$ a four dimensional, smooth Lorentzian manifold with $C^2$-metric $g$ and $\psi\in C^2(\mathcal{M})$, be a solution to the EKG system (\ref{EKG}).
Assume that $(\mathcal{M}, g, \psi)$ is invariant under an effective action of $SO(3)$ with principal orbit type an $S^2$. Denote by $r$ the area-radius of the spheres of symmetry. Then, locally around any point of $\mathcal{M}$, there exist double-null coordinates $u, v$ such that the metric takes the form
\ben{symmet}
g = -\Omega^2 du dv + r^2 d\sigma_{S^2}
\een
where $\Omega:=\Omega(u,v)$ and $r:=r(u,v)$ are $C^2$ functions\footnote{In fact, it suffices that $r \in C^2$, $\Omega\in C^1$ with $\Omega_{uv} \in C^0$ for the metric to be $C^2$ in the sense that the Riemann tensor has $C^0$ components.} and $d\sigma_{S^2}$ is the standard round metric of unit radius on $S^2$. Let $\mathcal{Q} = \mathcal{M} / SO(3)$ be the quotient of the spacetime by the isometry group. Then, the \EKG equations reduce to:
\bea
\partial_u \left(\frac{r_u}{\Omega^2} \right) &=& -{4 \pi r} \frac{(\partial_u \psi)^2}{\Omega^2}, \label{EKG1} \\
\partial_v \left(\frac{r_v}{\Omega^2} \right) &=& -{4 \pi r} \frac{(\partial_v \psi)^2}{\Omega^2}, \label{EKG2} \\
r_{uv} &=& - \frac{\Omega^2}{4 r} -\frac{r_u r_v}{r} + \frac{2 \pi a r}{l^2}\Omega^2 \psi^2 - \frac{3}{4} \frac{r}{l^2} \Omega^2,\label{EKG3} \\
\left( \log \Omega \right)_{uv} &=& \frac{\Omega^2}{4 r^2} + \frac{r_u r_v}{r^2} - 4 \pi \partial_u \psi \partial_v \psi, \label{EKG4}\\
\partial_u \partial_v \psi &=& - \frac{r_u}{r} \psi_v - \frac{r_v}{r}\psi_u - \frac{\Omega^2a}{2 l^2} \psi .\label{EKG5}
\eea
\end{Lemma}

While the variables $\Omega, r, \psi$ have a clean geometrical interpretation, they are not very suitable for the purposes of solving the system of equations \eq{EKG1}-\eq{EKG5} because they become singular at the conformal boundary of anti-de Sitter, where we expect $r\to \infty$, with $\Omega^2 \sim r^{2}, \psi \sim r^{-\frac{3}{2} + \kappa}$. In order to capture the asymptotic behaviour more carefully, we introduce a renormalised system of equations. We follow \cite{Holzegel:2011qk} in first introducing the Hawking mass:
\ben{hawkdef}
\varpi = \frac{r}{2}\left(1+ \frac{4 r_u r_v}{\Omega^2} \right) + \frac{r^3}{2 l^2}.
\een
This is a scalar under changes of $(u,v)$ coordinates which fix the metric form \eq{symmet} and is simply a constant for the Schwarzschild--anti-de Sitter metric. The Hawking mass obeys the following transport equations, which hold assuming \eq{EKG1}-\eq{EKG5}:
\bea
\partial_u \varpi &=& - 8 \pi r^2 \frac{r_v}{\Omega^2} (\partial_u \psi)^2 + \frac{4 \pi r^2 a}{l^2} r_u \psi^2 \, , \label{hawkueqn} \\
\partial_v \varpi &=& - 8 \pi r^2 \frac{r_u}{\Omega^2} (\partial_v \psi)^2 + \frac{4 \pi r^2 a}{l^2} r_v \psi^2  \, .\label{hawkveqn}
\eea
We can replace some of the \EKG system of equations in the previous Lemma with equations involving $\varpi$. For the purposes of the following Lemma, we may assume all derivatives to be taken in the weak sense.
\begin{Lemma}
Suppose that \eq{EKG3}, \eq{EKG5}, \eq{hawkueqn}, \eq{hawkveqn} hold, where $\Omega$ is understood to be defined by \eq{hawkdef}. Then as a consequence, \eq{EKG1}, \eq{EKG2} also hold. If furthermore the right hand side of \eq{EKG3} may be differentiated in $u$, then \eq{EKG4} holds.
\end{Lemma}
\begin{proof}
We first show that \eq{EKG1} holds as a consequence of \eq{EKG3}, \eq{hawkdef}, \eq{hawkueqn}. Consider the left hand side. We can replace $\frac{r_u}{\Omega^2}$ with a term involving $r, r_v, \varpi$ using \eq{hawkdef}. Differentiating this in $u$, we can replace the $r_{uv}$ and $\varpi_u$ terms which appear by making use of \eq{EKG3} and \eq{hawkueqn}. Simplifying the resultant expression, we arrive at \eq{EKG1}. Similarly \eq{EKG2} holds as a consequence of  \eq{EKG3}, \eq{hawkdef}, \eq{hawkveqn}. To show that \eq{EKG3} holds, we can multiply \eq{EKG1} by $\Omega^2$ and then differentiate with respect to $v$. Doing so, we obtain a term involving $(\log \Omega)_{uv}$, a term involving $\psi_{uv}$ and one involving $r_{uuv}$ together with lower order terms. The first of these we retain, the second can be replaced by making use of \eq{EKG5}, and the final one we can write as $\partial_u(r_{uv})$ and substitute in \eq{EKG3}. Simplifying the resulting expression, we arrive at \eq{EKG4}.
\end{proof}
The Hawking mass may loosely be thought of as the mass-energy inside a sphere of radius $r$. In the case of homogeneous Dirichlet conditions, this approaches a constant on the conformal boundary. For other choices of boundary condition, $\varpi$ in fact diverges towards the conformal boundary. This is a consequence of the fact that the un-renormalised energy in the scalar field $\psi$ is infinite for such boundary conditions. In the linear problem one must renormalise the energy-momentum tensor to give a finite energy for the field \cite{Warnick:2012fi,Holzegel:2012wt,BF}. In much the same way, we shall renormalise $\varpi$ and render it finite by subtracting a term which grows towards the boundary. To do so, we recall that key to the construction of the renormalised energy-momentum tensor for the linear problem is the introduction of twisted derivatives. Consider equation \eq{hawkueqn}. We can replace $\partial_u\psi$ with a twisted derivative as follows:
\be
\partial_u \psi = f \partial_u\left( \frac{\psi}{f} \right) + \frac{f_u}{f} \psi
\ee
for some $C^1$ function $f$. From here we deduce:
\begin{align}
\left(\partial_u \psi\right)^2 = \left(f \partial_u \left(\frac{\psi}{f}\right) + \psi \frac{f_u}{f}\right)^2 &= \left[ f \partial_u \left(\frac{\psi}{f}\right) \right]^2 + 2 f_u \psi  \partial_u \left(\frac{\psi}{f}\right) + \psi^2 \left[\frac{f_u}{f} \right]^2 \nonumber \\
&= \left[ f \partial_u \left(\frac{\psi}{f}\right) \right]^2 -  \psi^2 \left[\frac{f_u}{f} \right]^2 + \frac{f_u}{f} \partial_u \psi^2. \label{twpsisq}
\end{align}
Our intuition from the linear case leads us to expect that a suitable choice for $f$ is to take $f = r^{g}$ for some $g$ to be determined below. After substituting \eq{twpsisq} into \eq{hawkueqn}, the term involving $\partial_u \psi^2$ can be moved to the left hand side, at the expense of introducing some new zero'th order terms in $\psi$. Doing this and using \eq{hawkdef} to replace terms involving $r_u r_v /\Omega^2$, we find
\begin{align} \label{renhm}
\partial_u \left( \varpi -2\pi g \frac{r^3}{l^2} \psi^2\right)    =& -8\pi r^2 \frac{r_v}{\Omega^2}  \left[ f \partial_u \left(\frac{\psi}{f}\right) \right]^2+ 4\pi g \left(r-2\varpi\right) \psi \left( f \partial_u \frac{\psi}{f} \right) \nonumber \\
& +  2\pi \psi^2  r_u \left(\frac{r^2}{l^2} \left[-g^2-3g+2a\right] + g^2 \left(1-\frac{2\varpi}{r}\right) \right) 
\end{align}
Now we see that the choice $$g = -\frac{3}{2} + \kappa \, ,$$ suggested by linear theory, indeed leads to a cancellation of the top order term on the right hand side and we will henceforth work with $g$ defined by this choice. We therefore introduce a renormalised Hawking mass by
\ben{renhawk}
\varpi_N = \varpi - 2 \pi g \frac{r^3}{l^2}{\psi^2} \, .
\een
If \eq{EKG1}-\eq{EKG5} hold, then $\varpi_N$ obeys the equations
\begin{equation} \label{e3}
\begin{split}
\partial_v \varpi_N    =& -8\pi r^2 \frac{{r}_u}{{\Omega}^2}  \left[ f \partial_v \left(\frac{\psi}{f}\right) \right]^2+ 4\pi g \left(r-2\varpi_N\right) \psi \left( f \partial_v \frac{\psi}{f} \right) \\
& +  2\pi \psi^2  r_v \left( g^2 \left(1-\frac{2\varpi_N}{r}\right) \right) -16\pi^2g^2 \frac{r^3}{l^2} \psi^3 \left( f \partial_v \frac{\psi}{f} \right) - 8\pi^2 g^3 \frac{r^2}{l^2} r_v \psi^4
\end{split}
\end{equation}
\begin{equation} \label{e4}
\begin{split}
\partial_u \varpi_N    =& -8\pi r^2 \frac{{r}_v}{{\Omega}^2}  \left[ f \partial_u \left(\frac{\psi}{f}\right) \right]^2+ 4\pi g \left(r-2\varpi_N\right) \psi \left( f \partial_u \frac{\psi}{f} \right) \\
& +  2\pi \psi^2  r_u \left( g^2 \left(1-\frac{2\varpi_N}{r}\right) \right) -16\pi^2g^2 \frac{r^3}{l^2} \psi^3 \left( f \partial_u \frac{\psi}{f} \right) - 8\pi^2 g^3 \frac{r^2}{l^2} r_u \psi^4
\end{split}
\end{equation}
which follow immediately from \eq{renhm}, together with the same equation after swapping $u, v$.

To renormalise the wave equation for $\psi$ \eq{EKG5}, we can simply follow the procedure applied in \cite{Warnick:2012fi,Holzegel:2012wt} for twisting a Klein-Gordon equation. We claim that by expanding the terms  (assuming $f=r^{-\frac{3}{2}+\kappa}$, $r\in C^1$, and that the equation \eq{EKG3} for $r_{uv}$ holds) the following equations are readily seen to be equivalent to one another and also to \eq{EKG5}.
\bea
\partial_v \left(f r \left(\partial_u \frac{{\psi}}{f}\right)\right) &=& - \partial_u \left(rf\right) \left(\partial_v \frac{{\psi}}{f}\right) - \frac{\Omega^2}{4} r V{\psi}, \label{e2} \\
\partial_u \left(f r \left(\partial_v \frac{{\psi}}{f}\right)\right) &=& - \partial_v \left(rf\right) \left(\partial_u \frac{{\psi}}{f}\right) - \frac{\Omega^2}{4} r V{\psi}, 
\eea
where the potential is given by:
\begin{align} \label{e2b}
V &= \frac{2\varpi}{r^3} \left(\kappa - \frac{3}{2}\right)^2 + 8\pi \left(\kappa - \frac{3}{2}\right) \frac{a}{l^2} \psi^2 - \frac{1}{r^2} \left(\kappa^2 - 2\kappa + \frac{3}{4}\right) \nonumber \\
&= \frac{2\varpi_N}{r^3} \left(\kappa - \frac{3}{2}\right)^2 + \frac{\psi^2}{l^2}\left(\kappa-\frac{3}{2}\right) \left[8\pi a +4\pi \left(\kappa-\frac{3}{2}\right)^2 \right]  - \frac{1}{r^2} \left(\kappa^2 - 2\kappa + \frac{3}{4}\right) \, .
\end{align}
Note that for $\kappa>\frac{1}{2}$, i.e.\ beyond the conformally coupled case $a=-1$, the potential decays \emph{slower} than $r^{-2}$. This is a consequence of the fact that, even assuming all the metric functions are smooth on the interior, the rescaled metric $r^{-2} g$ can no longer be extended as a $C^2$ metric across the conformal boundary, but rather only in $C^{1, 2-2\kappa}$. We shall be forced to confront this issue at various points in our arguments.

Finally, the radial coordinate $r$ may be simply renormalised by considering instead $\rt = \frac{1}{r}$. Making use of \eq{EKG3}, together with the expressions \eq{hawkdef}, \eq{renhawk} relating $\Omega, \varpi_N$, it is a matter of simple calculation to show that
\begin{align} \label{e1}
\rt_{uv} = 
  \frac{\Omega^2}{r^2} \left(\frac{3 \varpi_N}{2 r^2} -\frac{1}{2r} + \frac{2 \pi  r \psi^2}{l^2} \left(-a+\frac{3}{2}g\right)\right).
\end{align}

\subsection{Notation} 
In view of their importance, we introduce a notation for the twisted derivatives introduced above. We let $\rho=\frac{1}{2}\left(u-v\right)$ and define
\begin{equation} \label{htd2}
\hat{\partial}_u \psi := \rho^{\frac{3}{2}-\kappa}\frac{\partial}{\partial u} \left (\frac{\psi}{\rho^{\frac{3}{2}-\kappa}} \right)
\end{equation}
and
\begin{equation} \label{ttd}
\tilde{\partial}_u \psi := \rt^{\frac{3}{2}-\kappa}\frac{\partial}{\partial u} \left (\frac{\psi}{\rt^{\frac{3}{2}-\kappa}} \right) \, .
\end{equation}
Note that in (\ref{htd2}) we twist with the function $\rho$ known explicitly in terms of the coordinates $u$ and $v$, while in (\ref{ttd}) we twist with the geometric area radius $\tilde{r}$ which is itself a dynamical variable. The wave equation (\ref{e2}) twists naturally with $\tilde{r}$ while the norms are more cleanly expressed in terms of $\rho$-twisted derivatives. A relation between (\ref{htd2}) and (\ref{ttd}) in the context of the contraction map is established in Lemma \ref{lem:twisteq}.

We also denote $t=\frac{1}{2}\left(u+v\right)$ and observe that this is a useful coordinate along $\scri$.

\subsection{Restriction on $\kappa$} \label{sec:restrict}
Let us recapitulate what the above renormalization has achieved. We recall that from the linear theory we expect
\[
\psi \sim r^{-3/2+\kappa} \ \ \ , \ \ \  \tilde{\partial}_u{\psi}, \tilde{\partial}_v \psi \sim  r^{\max\left(-3/2+\kappa, -\frac{1}{2}-\kappa\right)} \ \ \ , \ \ \ \partial_u \psi , \partial_v \psi \sim  r^{-1/2+\kappa} \, .
\] 
Investigating the right hand side of (\ref{e3}) and (\ref{e4}) we see that (assuming the decay from the linear theory for the moment) all but the last of the five terms are integrable for $0<\kappa<1$, while the last is integrable only for $0<\kappa < \frac{3}{4}$. This situation can be remedied with an \emph{additional renormalization} to be discussed in Section \ref{sec:kap}. A further restriction on $\kappa$, $\kappa<\frac{2}{3}$, will arise when proving the energy estimate for the wave equation (\ref{e2}) in view of the $\psi^2$-term in the potential (\ref{e2b}) not decaying sufficiently strongly. This can also be remedied as shown in Section \ref{sec:kap}. However, to avoid cumbersome formulae and obscuring the main ideas, for the remainder of the paper we are simply going to assume 
\begin{align} \label{massrestrict}
0< \kappa < \frac{2}{3} \, .
\end{align}
In Section \ref{sec:kap} we outline a proof of the general case $0<\kappa<1$. 

It is not surprising that the problem becomes more technically challenging for $\kappa$ close to $1$. The solutions we shall construct at the $H^1$ level will have an expansion in a suitable coordinate chart of the form
 \bean
 \psi &=& \psi^-(t) \rho^{\frac{3}{2} - \kappa} + \O{\rho^{\frac{3}{2}}} \\
 g &=& \frac{l^2}{\rho^2} \left[ \left(1+\O{\rho^\eta} \right ) du dv +  \left(1+\O{\rho^\eta} \right ) d\sigma_{S^2}   \right]
 \eean
 where $\eta = \min(2, 3-2\kappa)$. Moreover this expansion is sharp: at the classical level of regularity one indeed sees terms in the metric proportional to $\rho^{3-2\kappa}$ and $\rho^2$ which cannot be removed by a coordinate choice. We thus see that the metric is only weakly asymptotically AdS for $\kappa >\frac{1}{2}$.

\subsection{The renormalised problem}
Motivated by the previous considerations, we are now ready to set up the problem which we will actually solve. Define the triangle $\Delta_{\delta,u_0} := \{(u, v) \in \R^2: u_0 \leq v \leq u_0+\delta, v<u \leq u_0+\delta \}$, and the conformal boundary $\scri := \overline{\Delta}_{\delta, u_0} \setminus \Delta_{\delta,u_0} = \{(u, v) \in \overline \Delta_{\delta,u_0} : u=v\}  $.  We shall allow ourselves to write $\Delta$ for $ \Delta_{\delta,u_0}$ as long as there is no ambiguity. We will take as our dynamical variables
\begin{align}
\tilde{r} : \Delta_{\delta,u_0} \longrightarrow \mathbb{R}^+ \ \ \ , \ \ \ \psi : \Delta_{\delta,u_0} \longrightarrow \mathbb{R} \ \ \ , \ \ \ \varpi_N : \Delta_{\delta,u_0} \longrightarrow \mathbb{R} \, ,
\end{align}
and treat these as \emph{defining} the auxiliary variables:
\ben{defaux}
r:=\frac{1}{\tilde{r}},  \quad \varpi := \varpi_N + 2\pi g \frac{r^3}{l^2} \psi^2,  \quad  \Omega^2 := -\frac{4r^4 \tilde{r}_u \tilde{r}_v}{1-\mu} \, , \quad 1-\mu := 1-\frac{2 \varpi}{r} + \frac{r^2}{l^2}.
\een
With these definitions, we can understand   \eq{e3}, \eq{e4},  \eq{e2}, \eq{e1} as equations for $\rt, \varpi_N, \psi$.

\begin{Definition} \label{def:wsEKG}
A weak solution to the renormalised \EKG equations is a triple $(\rt, \varpi_N, \psi)\in C^1_{loc.} \cap W^{1,1}_{loc.} \cap H^1_{loc.}$ such that $\psi_u, \rt_{uu}, (\varpi_N)_u\in C^0_{loc.}$ and which satisfies   \eq{e3}, \eq{e4},  \eq{e2}, \eq{e1} in a weak sense.
\end{Definition}
We note that as a consequence of the equations holding, a weak solution to the renormalised \EKG equations necessarily has $\rt_{uv}, \rt_{uuv}\in C^0_{loc.}$.
We justify considering the renormalised system of equations with the following result.
\begin{Lemma}
Suppose that we have a weak solution to the renormalised \EKG equations. Then in fact the equations \eq{EKG1}-\eq{EKG5} hold in a weak sense, and hence we can say that the metric \eq{symmet} satisfies the \EKG equations, \eq{EKG}, in a weak sense.
\end{Lemma}
\begin{Remark}
Such a statement obviously makes sense with higher regularity. In particular if $r\in C^2_{loc.}, \varpi \in C^1_{loc.}, \psi \in C^1_{loc.}$ then the metric $g$ defined by \eq{symmet} has $C^0$ curvature, and the \EKG equations hold in a classical sense. We also remark that (\ref{e4}) only needs to hold on the initial data and is then propagated by (\ref{e1}), (\ref{e2}), (\ref{e3}) as shown explicitly in Section  \ref{sec:constraintprop}.
\end{Remark}

\section{Initial and boundary data} \label{initdatasec}

\subsection{Initial data}
In this section we shall give conditions on initial data which are sufficient for the construction of a weak solution to the \EKG system. When we turn later to showing that better regularity is propagated by the equations, we shall introduce further conditions, see Section \ref{sec:improreg}.

\begin{Definition} \label{def:id}
Let $\mathcal{N}=\left(u_0,u_1\right]$ be a real interval. We call a pair of functions
$\left(\overline{\tilde{r}}, \overline{\psi}\right) \in C^{2} \left(\mathcal{N}\right) \times  C^{1} \left(\mathcal{N}\right)$ a \underline{free data set}, provided the following holds:
\begin{itemize}
\item $\overline{\tilde{r}} > 0$ and $\overline{\tilde{r}}_u > 0$ in $\mathcal{N}$, as well as $\lim_{u \rightarrow u_0} \overline{\tilde{r}}\left(u\right) = 0$, $\lim_{u \rightarrow u_0} \overline{\tilde{r}}_u\left(u\right) = \frac{1}{2}$ and $\lim_{u \rightarrow u_0} \overline{\tilde{r}}_{uu}=0$.
\item There is a constant $C_{data}$ such that 
\begin{align} \label{asp1}
\int_{u_0}^{u_1} \left[ \left(\overline{f} \partial_u \left( \frac{\overline{\psi}}{\overline{f}}\right)\right)^2  + \overline{\psi}^2\right] \left(u-u_0\right)^{-2} du    <C_{data}
\end{align}
\begin{align} \label{asp2}
 \sup_{\mathcal{N}} | \overline{\psi} \cdot \overline{\tilde{r}}^{-\frac{3}{2}+\kappa} | +  \sup_{\mathcal{N}} \Big| r^{\frac{1}{2}} \left( \overline{f} \partial_u\frac{\overline{\psi}}{\overline{f}}\right) \Big| < C_{data}
\end{align}
Here $\overline{f}= \left[\frac{1}{2}\left(u-u_0\right)\right]^{3/2-\kappa}$.
In particular, the limit $\Psi:=\lim_{u \rightarrow u_0} \overline{\psi} \cdot \overline{\tilde{r}}^{-\frac{3}{2}+\kappa}$ exists.
\end{itemize} 
\end{Definition}

From a free data set as above, we construct a complete initial data set $\left(\overline{\tilde{r}}, \overline{\psi}, \overline{\varpi_N}, \overline{\tilde{r}_v}\right) \in C^{2} \left(\mathcal{N}\right) \times  C^{1} \left(\mathcal{N}\right) \times C^{1} \left(\mathcal{N}\right) \times C^{1} \left(\mathcal{N}\right) $ by integrating the constraints as follows.

The function $\overline{\varpi_N}$ is obtained  as the unique solution $\overline{\varpi_N} \in C^{1} \left(\mathcal{N}\right)$ of the linear ODE 
\begin{equation}
\begin{split} \label{varpied}
\left(\overline{\varpi_N}\right)_u = 2\pi \overline{r}^2 \frac{1-\frac{2\varpi_N}{\overline{r}} + \frac{\overline{r}^2}{l^2} -4\pi g \frac{\overline{r}^2}{l^2}\overline{\psi}^2}{\overline{r}_u}  \left[ f \partial_u \left(\frac{\overline{\psi}}{f}\right) \right]^2+ 4\pi g \left(\overline{r}-2\overline{\varpi_N}\right) \overline{\psi} \left( f \partial_u \frac{\overline{\psi}}{f} \right) \\
 +  2\pi \overline{\psi}^2  \overline{r}_u \left( g^2 \left(1-\frac{2\overline{\varpi_N}}{\overline{r}}\right) \right) -16\pi^2g^2 \frac{\overline{r}^3}{l^2} \overline{\psi}^3 \left( f \partial_u \frac{\overline{\psi}}{f} \right) - 8\pi^2 g^3 \frac{\overline{r}^2}{l^2} \overline{r}_u \overline{\psi}^4
 \end{split}
\end{equation}
where $\overline{r}:=\overline{\tilde{r}}^{-1}$, $\overline{r}_u=\frac{\overline{\tilde{r}}_u}{\overline{\tilde{r}}^2}$ corresponds to the original geometric area radius function,
subject to the boundary condition
\begin{align} \label{bco}
\lim_{u \rightarrow u_0} \overline{\varpi_N} = M_N \, .
\end{align}
The function $\overline{\tilde{r}_v}$ is obtained as the unique solution $\overline{\tilde{r}_v} \in \overline{\varpi_N} \in C^{1} \left(\mathcal{N}\right)$ of the ODE
\begin{align} \label{rvinitial}
\left(\overline{\tilde{r}_v}\right)_u = \frac{-4\overline{r}^2\overline{\tilde{r}}_u \overline{\tilde{r}_v}}{1-\frac{2\varpi_N}{\overline{r}} + \frac{\overline{r}^2}{l^2} -4\pi g \frac{\overline{r}^2}{l^2}\overline{\psi}^2} \left(\frac{3 \overline{\varpi_N}}{2 \overline{r}^2} -\frac{1}{2\overline{r}} + \frac{2 \pi  \overline{r} \overline{\psi}^2}{l^2} \left(-a+\frac{3}{2}g\right)\right)
\end{align}
with boundary condition
\begin{align}
\lim_{u \rightarrow u_0} \overline{\tilde{r}_v} = -\frac{1}{2} \, .
\end{align}

\begin{Remark}
The choice of $\overline{\tilde{r}}$ fixes the scale of the $u$-coordinate along $\mathcal{N}$ corresponding to the gauge-freedom in the problem. A simple and convenient choice is $\overline{\tilde{r}}=\frac{1}{2}\left(u-u_0\right)$. The function $\overline{\psi}$ is the free data in the problem and can be specified arbitrarily modulo the integrability conditions of Definition \ref{def:id}.

The choice of boundary condition for $\overline{\tilde{r}_v}$ ensures that initially $T\overline{\tilde{r}}=0$ corresponding to the fact that we would like to have $T\tilde{r}=0$ along the boundary $u=v$ in the evolution. It is also a convenient gauge freedom.

The choice of boundary condition for $\overline{\varpi_N}$ is again ``free". However, we could also specify an initial value at $u_1$ and integrate outwards, determining $\overline{\varpi_N}$ as $u \rightarrow u_0$, which may be the case in applications where $u_1$ corresponds to the axis on which a regularity condition $\varpi_N=0$ has to be imposed.
\end{Remark}

The following Lemma is useful and a direct consequence of Definition \ref{def:id}.
\begin{Lemma} \label{lem:databounds}
For any $0<s<1$, given $\delta^\prime>0$ we can choose $\delta>0$ such that the following bounds hold on the truncated initial data ray $\mathcal{N}_\delta := \mathcal{N} \cap \{u \leq u +\delta\}$:
\begin{align} \label{sr1}
\| \overline{\tilde{r}} \|_{C^0} + \| \overline{\tilde{r}}_u - \frac{1}{2} \|_{C^0}  + \| \overline{\tilde{r}_v} + \frac{1}{2} \|_{C^0} + \| \overline{\tilde{r}}_{uu} \|_{C^0} < \delta^\prime
\end{align}
\begin{align} \label{sr2}
\int_{u_0}^{u_0+\delta} \left[ \left(\overline{\tilde{r}}^{-1} \cdot \overline{f} \partial_u \left( \frac{\overline{\psi}}{\overline{f}}\right)\right)^2  + \overline{\tilde{r}}^{-2} \overline{\psi}^2\right] du    < \delta^\prime    \ \ \ \textrm{and} \ \ \ \Big| \overline{\tilde{r}}^{-\frac{1}{2}+\frac{s}{4}} \cdot \overline{f} \partial_u \left( \frac{\overline{\psi}}{\overline{f}}\right)\Big|<\delta^\prime
\end{align}
\begin{align} \label{sr3}
\| \overline{\varpi_N}-M_N \|_{C^0} < \delta^\prime   \ \ \ \textrm{and} \ \ \  \| \overline{\tilde{r}}^{1+s} \partial_u\overline{\varpi_N} \|_{C^0} < \delta^\prime
\end{align}
\begin{align} \label{sr4}
 \| \bar{\psi} \bar{\rho}^{-\frac{3}{2}+\kappa} - \Psi \|_{C^0} < \delta^\prime
\end{align}
where $\| \cdot \|_{C^0}$ denotes the $sup$-norm in $\mathcal{N}_\delta$.
\end{Lemma}

\begin{Remark}
The appearance of $s$ is merely technical (to guarantee an additional smallness factor). The weights could be improved in the context of higher regularity. In particular, one expects to be able to propagate sharper decay for $\overline{f} \partial_u \left(\frac{\overline{\psi}}{\overline{f}}\right)$ if higher ($C^2$ regularity of $\psi$) is imposed.
\end{Remark}

\begin{proof}
The bound (\ref{sr1}) follows from the fact that $\overline{\tilde{r}}$ is $C^2$ and its asymptotics at $(u_0,v_0)$. The first bound of (\ref{sr2}) follows from localizing (\ref{asp1}). Using (\ref{sr2}), integrating the equation (\ref{varpied}) for $\left(\overline{\varpi_N}\right)_u$ establishes the first bound of (\ref{sr3}) after carefully checking the $\overline{\tilde{r}}$-weights in each term. The second bound of (\ref{sr2}) follows directly from (\ref{asp2}) using $C_{data} \overline{\tilde{r}}^{s/4} \leq C_{data} \delta^{s/4}<\delta^\prime$. The second bound of (\ref{sr3}) follows from estimating pointwise the right hand side of (\ref{varpied}) after multiplying it by $\overline{\tilde{r}}^{1+s}$. The bound (\ref{sr4}) follows from
\begin{align} \label{sty}
 | \overline{\psi} \bar{\rho}^{-\frac{3}{2}+\kappa} -\Psi  | \leq 0 + \int_{u_0}^{u_0+\delta} du^\prime |\partial_u \left( \overline{\psi} \bar{\rho}^{-\frac{3}{2}+\kappa} -\Psi \right)| \leq \|\overline{\psi}\|_{\underline{H}^1 \left(u_0,u_0+\delta\right)} \delta^{\kappa}  \, .
\end{align}
\end{proof}

\subsection{Boundary Conditions} \label{sec:bdycondition}

We require boundary conditions for the fields in order to produce a unique evolution. For the dynamical field $\psi$ there are a variety of boundary conditions studied in the context of the linear problem in \cite{Warnick:2012fi, Holzegel:2012wt}. We shall work with the non-linear version of inhomogeneous Robin conditions, which includes the homogeneous Neumann boundary condition as a special choice. While we do not discuss the inhomogeneous Dirichlet condition, it can be treated by precisely the same methods. We will state the boundary conditions on $\psi$ in a form that may be applied to the non-spherically symmetric case, before specialising to the case in hand.

We say that a triple $(\uprho, \upbeta, \upgamma)$ is a representative choice of boundary data if $\uprho$ is a smooth boundary defining function for $\scri$ (i.e.\ $\uprho>0$ on $\Delta \setminus \scri$, with $\uprho = 0, d\uprho \neq 0$ on $\scri$) and $\upbeta, \upgamma$ are functions along $\scri$. We will take $\upbeta, \upgamma$ to be smooth, but this is stronger than required. Given a representative choice of boundary data, we define $P$ to be the unique vector which is normal to $\scri$ with respect to the rescaled metric $\uprho^2 g$ and further satisfies $P(\uprho) = 1$. We say that $\psi$ satisfies the boundary conditions determined by $(\uprho, \upbeta, \upgamma)$ if
\ben{bcdef}
\left. \uprho^{1-2\kappa} P\left (\uprho^{\kappa-\frac{3}{2}} \psi\right ) \right|_{\scri} +  2 \left. \left (\uprho^{\kappa-\frac{3}{2}} \psi\right ) \right|_{\scri} \upbeta = \upgamma.
\een
Notice that if $\omega$ is a smooth function with $\omega>0, P(\omega)=0$ on $\scri$, then the representative choice of boundary data $(\omega \uprho, \omega^{1-2\kappa}\upbeta, \omega^{-\frac{1}{2}-\kappa}\upgamma)$ gives rise to the same boundary conditions. If $\kappa<\frac{1}{2}$ then the requirement on $P(\omega)$ may be dropped. We define a choice of boundary data $\mathcal{B} = [(\uprho, \upbeta, \upgamma)]_\sim$ to be an equivalence class of representative choices of boundary conditions under the equivalence relation
\be
(\uprho, \upbeta, \upgamma) \sim (\omega \uprho, \omega^{1-2\kappa}\upbeta, \omega^{-\frac{1}{2}-\kappa}\upgamma), \qquad \omega \in C^\infty; \ \omega>0\textrm{ and } P(\omega)=0 \textrm{ on } \scri
\ee
Here we understand that for $\kappa<\frac{1}{2}$ the condition $P(\omega)$ can be dropped. Note that the homogeneous Neumann boundary conditions $\upbeta=\upgamma=0$ are invariant under the similarity transformation, so for these boundary conditions the choice of boundary defining function is immaterial.

In this paper, we will work with boundary conditions of the form $ [(\rho, \beta, \gamma)]_\sim$, where $\rho = \frac{1}{2}(u-v)$ was previously introduced. For $\kappa\leq \frac{1}{2}$ this represents no restriction, while for $\kappa > \frac{1}{2}$ there exist choices of boundary data which do not belong to this set.  

For technical reasons, it will turn out to be very convenient to work with the boundary conditions in the form
\begin{align} \label{bcpsi}
\rho^{-\frac{1}{2}-\kappa} \left(\tilde{\partial}_u - \tilde{\partial}_v \right) \psi + 2\beta\left(t\right) \rho^{-\frac{3}{2} + \kappa} \psi = \gamma\left(t\right) \ \ \  \textrm{on $\scri$}
\end{align}
where $\tilde{\partial}$ is the derivative twisted with respect to $r$. These two conditions can be seen to be equivalent provided that
\be
\rho^{1-2\kappa}\left (\rho^{\kappa-\frac{3}{2}} \psi\right ) \left[ \frac{\rt_u - \rt_v }{\rt} -\frac{1}{\rho}\right] \to 0
\ee
as $\scri$ is approached. The term in square brackets can be shown to be bounded for solutions at the $H^1$ level of regularity, which gives equivalence of \eq{bcdef}, \eq{bcpsi} for $\kappa<\frac{1}{2}$. At the $H^2$ level, the term in square brackets has improved asymptotics of $\O{\rho^{\min(1, 2-2\kappa)}}$, which shows equivalence for $\kappa<\frac{3}{4}$. The reason we need to improve regularity to show equivalence seems to be that for $\kappa \geq \frac{1}{2}$ one requires cancellations coming from the next to leading order terms in the expansion of $\psi$ near infinity. At the $H^1$ level one has no control over these in general, but certain combinations (such as $\psi \rt^{-\frac{3}{2}+\kappa}$) exhibit better behaviour than one may expect. At the $H^2$ level of regularity another term in the expansion is available with which one can see cancellations explicitly. 
\begin{Remark}
As is well-known, the boundary condition (\ref{bcpsi}) does not make sense classically if $\psi$ is only in $H^1$. See the paper \cite{Warnick:2012fi} for the appropriate weak formulations.
\end{Remark}

One might wonder why we introduce a boundary defining function $\uprho$, rather than stating the boundary conditions in terms of the geometric quantity $\rt$ which furnishes a convenient, canonical,  boundary defining function. We could take as boundary conditions:
\ben{geombc}
\left. \rt^{1-2\kappa}P \left (\rt^{\kappa-\frac{3}{2}} \psi\right ) \right|_{\scri} +  2 \left. \left (\rt^{\kappa-\frac{3}{2}} \psi\right ) \right|_{\scri} \beta = \gamma.
\een
with $P(\rt)=1$. This is not included in our choice of boundary data allowed above, since it assumes knowledge of $\rt$ which we do not have until we have found the solution. Our reasons for not considering these boundary conditions are twofold. Firstly, the existence of $\rt$ is a feature of the spherical symmetry. With the non-spherically symmetric problem in mind it is clear that the boundary conditions may only be stated once one has made a choice of $\uprho$. The second reason is a technical one: namely that for $\kappa \geq \frac{1}{2}$ we cannot, unless $\beta=0$, close the contraction map argument with these boundary conditions at the $H^1$-level. However, we believe that the problem with boundary conditions \eq{geombc} could also be solved directly for $\kappa \geq \frac{1}{2}$ by closing the contraction map at the $H^2$-level.

We shall also require some boundary conditions for the metric. In spherical symmetry this reduces to a condition on $\rt$. To produce  aAdS spacetimes we impose 
\begin{equation} \label{bcr}
\tilde{r}|_{\scri} = 0 \, .
\end{equation}

A consequence of our choice of boundary conditions is that the renormalised Hawking mass at infinity, which we may think of as a measure of the energy in the spacetime, need not be constant. To state the properties of the renormalised Hawking mass at the boundary cleanly, it is convenient to introduce two vector fields which are invariant under changes of the $u, v$ coordinates preserving the form of the metric. These are $\mathcal{T} = \Omega^{-2}(r_u \partial_v - r_v \partial_u)$ and $\mathcal{R} = -\Omega^{-2}(r_u \partial_v + r_v \partial_u)$. Examining the fall-off of the terms in the $\varpi_N$ evolution equations, we find that if \eq{e3}, \eq{e4} are satisfied then
\be
\left. \mathcal{T} \varpi_N \right |_{\scri} = \lim_{\rho \to 0}  8 \pi r^2 (\tilde{\mathcal{T}} \psi) (\tilde{\mathcal{R}} \psi)
\ee
Where $\tilde{\mathcal{T} }\psi :=\mathcal{T}^\mu \tilde{\nabla}_\mu \psi$, and similarly for the other derivative. The right hand side has a finite limit in $L^1(\scri)$ provided that $\psi$ is at least $\H^2$, from the results of \cite{Warnick:2012fi} (see \S\ref{sec:improreg}). Notice that for homogeneous Dirichlet conditions (corresponding to $\tilde{\mathcal{T}} \psi=0$) or homogeneous Neumann ($\tilde{\mathcal{R}} \psi=0$), the renormalised Hawking mass is conserved. Otherwise we find that the time derivative is proportional to the energy flux of the field $\psi$ across $\scri$, as one might expect.

\section{The Main Theorem} \label{sec:maintheo}

We are now ready to state the main theorem.
\begin{Theorem} \label{theo:main}
Fix $0<\kappa<2/3$ and
let $\left(\overline{\tilde{r}} ,\overline{\psi}\right)$ be a free data set on the interval $\mathcal{N}=\left(u_0,u_1\right]$ as defined in Definition \ref{def:id}. Fix also a choice of boundary condition of the form (\ref{bcpsi}), where $\beta$ and $\gamma$ are smooth along $\scri$. Then there exists a $\delta>0$ such that the following holds. There exists a unique weak solution $\left(\tilde{r},\varpi_N, \psi\right)$ of the renormalised \EKG equations (cf.~Definition \ref{def:wsEKG}) in the triangle $\Delta_{\delta,u_0}$ such that 
\begin{itemize}
\item $\tilde{r}$ satisfies (\ref{e1}) with boundary condition (\ref{bcr}) 
\item $\psi$ satisfies (\ref{e2}) with boundary condition (\ref{bcpsi}) in a weak sense 
\item The functions $\psi$ and $\tilde{r}$ agree as $C^1$ functions with $\overline{\psi}$ and $\overline{\tilde{r}}$ respectively when restricted to $u=u_0$.
\end{itemize}
\end{Theorem}
\begin{proof}
The results that prove this theorem make up \S\ref{proofmaintheorem}. The solution is constructed by a fixed point argument for a map $\Phi$, constructed in \S \ref{contractionsec}. Propositions \ref{prop:maptoball}, \ref{prop:conmap}, assert that $\Phi$ is a contraction map and Corollary \ref{maincor} then asserts the existence of a unique weak solution  to \eq{e3},  \eq{e2}, \eq{e1} with given intial-boundary data. Finally, Corollary \ref{constcor} asserts that for such a solution, the constraint equation \eq{e4} propagates in the evolution.
\end{proof}

\begin{Remark}
The restriction on $\kappa$ is technical and could be improved to the full range $0<\kappa<1$ with an additional renormalization of the system of equations. See  \S \ref{sec:restrict} and  \S \ref{sec:kap}. The theorem may be extended to consider nonlinear potentials for the Klein-Gordon equation, as well as nonlinear boundary conditions. These possibilities are discussed in \S \ref{sec:nonlinear} and \S \ref{sec:nlbc}.
\end{Remark}

Given a weak solution we can improve the regularity and in particular obtain a classical solution:
\begin{Theorem}\label{highregthm}
Suppose the initial data of Theorem \ref{theo:main} satisfy the additional regularity conditions of Section \ref{sec:improreg} then the weak solution is actually a classical solution.
\end{Theorem}
\begin{proof}
This follows immediately from Corollary \ref{highregcor}, established in \S\ref{sec:improreg}.
\end{proof}

\noindent \textbf{Geometric Uniqueness.} A priori, Theorem \ref{theo:main} only provides a uniqueness statement in the double-null coordinates in which the theorem is proven. For homogeneous Neumann boundary conditions, one can define the notion of a maximum development and obtain also a \emph{geometric} uniqueness statement within spherical symmetry. This argument follows precisely \S8.1 of \cite{Holzegel:2011qk}.

For Robin conditions, as well as inhomogeneous conditions it appears that a geometric uniqueness result of this kind does not hold. The reason is that one requires a choice of boundary defining function $\uprho$ in order to state such boundary conditions, and a choice of $\uprho$ necessarily makes reference to the spacetime manifold itself (rather than being intrinsic to the embedded surface $\scri$). In this circumstance, we may say that for a given spacetime manifold with initial data $(\overline{\rt}, \overline{\psi})$ and boundary data $[(\uprho, \upbeta, \upgamma)]_\sim$ specified, the fields $g, \psi$ are uniquely determined in the domain of dependence of the data. This is weaker than the geometric uniqueness statement for homogeneous Dirichlet or Neumann boundary conditions, which may crudely be thought of as asserting the uniqueness of the spacetime manifold itself, given the initial data.

\section{Proof of Theorem \ref{theo:main}}\label{proofmaintheorem}
\subsection{The function spaces} \label{funspacessec}

We set up the appropriate function spaces for the dynamical variables.
We denote by $C^{1+}_{\tilde{r}} \left(\Delta_{\delta,u_0}\right)$ the space of positive functions $\tilde{r}$ on $\Delta_{\delta,u_0}$ that are $C^1$ in $\Delta_{\delta,u_0}$, agree with $\overline{\tilde{r}}$ on $\mathcal{N}$ and are such that both the $uv$- and the $uu$-derivative exist and are continuous. We employ that space with the distance:
\begin{align}
d_{\tilde{r}} \left(\tilde{r}_1,\tilde{r}_2\right) = \|\log \frac{\tilde{r}_1}{\tilde{r}_2} \|_{C^0} + \|\log |(\tilde{r}_1)_u| - \log |(\tilde{r}_2)_u| \|_{C^0} +  \|\log |(-\tilde{r}_1)_v| - \log |(-\tilde{r}_2)_v| \|_{C^0} \nonumber \\
+\Big\| \frac{T(\tilde{r}_1)}{\rho} - \frac{T(\tilde{r}_2)}{\rho}\Big\|_{C^0} +\|\left(\tilde{r}_1\right)_{uv} -\left(\tilde{r}_2\right)_{uv} \|_{C^0} + \|\left(\tilde{r}_1\right)_{uu} -\left(\tilde{r}_2\right)_{uu}\|_{C^0} \, .
\end{align}
Here $\|\cdot \|_{C^0}=\sup_{\Delta_{\delta,u_0}} | \cdot |$ denotes the $\sup$-norm in the triangle $\Delta_{u_0,\delta}$.
Similarly, we define $C^{0+}_{\varpi_N} \left(\Delta_{\delta,u_0}\right)$ as the space of real-valued functions $\varpi_N$ that are $C^0$ in $\Delta_{\delta,u_0}$, agree with $\overline{\varpi_N}$ on $\mathcal{N}$ and are such that the $u$-derivative exists and is continuous. We equip that space with the distance
\[
d_{\varpi} \left(\left(\varpi_N\right)_1, \left(\varpi_N\right)_2\right) = \|  \left(\varpi_N\right)_1 - \left(\varpi_N\right)_2 \|_{C^0} + \| \rho^{1+s} \partial_u \left(\varpi_N\right)_1 - \rho^{1+s}\partial_u \left(\varpi_N\right)_2 \|_{C^0} \, .
\]
 The appearance of the small number $0<s<1$ is technical and will provide an additional source of smallness in the contraction map.
Finally, $C^{0+}_\psi \underline{H}^1\left(\Delta \right)$ is the space of real-valued functions that are continuously differentiable in $u$, agree with $\overline{\psi}$ on $\mathcal{N}$ and are both continuous in $u$ with values in $\underline{H}^1\left(v\right)$ and continuous in $v$ with values in $\underline{H}^1\left(u\right)$. We equip that space with the distance
\[
d_{\psi}\left( \psi_1, \psi_2 \right) = \| \psi_1-\psi_2 \|_{C^0\underline{H}^1} +  \| \psi_1 \rho^{-\frac{3}{2}+\kappa} -\psi_2 \rho^{-\frac{3}{2}+\kappa} \|_{C^0} + \|\rho^{-\frac{1}{2}+\frac{s}{4}} \hat{\partial}_u \psi_1 - \rho^{-\frac{1}{2}+\frac{s}{4}} \hat{\partial}_u \psi_2\|_{C^0} \, ,
\]
where we recall the definition of the twisted derivative (\ref{htd2}) 
and the norm $(\rho := \frac{u-v}{2})$:
\begin{align} \label{c0h1}
\|\psi\|^2_{C^0 \underline{H}^1\left(\Delta \right)} =  \sup_{(u,v) \in \Delta} \int_v^u \left[ \rho^{-2} \left(\hat{\partial}_u \psi\right)^2 + \rho^{-2} \psi^2 \right] du^\prime \nonumber \\ +  \sup_{(u,v) \in \Delta} \int_{v_0}^v \left[ \rho^{-2} \left(\hat{\partial}_v \psi\right)^2 + \rho^{-2} \psi^2 \right] dv^\prime \, .
\end{align}

This produces the complete metric space $\mathcal{C} = C^{1+}_{\tilde{r}} \left(\Delta_{\delta,u_0}\right)\times C^{0+}_{\varpi_N} \left(\Delta_{\delta,u_0}\right) \times C^{0+}_{\psi} \underline{H}^1\left(\Delta_{\delta,u_0} \right)$ with distance
\[
d\left( \left(\tilde{r}_1,(\varpi_N)_1, \psi_1\right) , \left(\tilde{r}_2,(\varpi_N)_2, \psi_2\right)\right) = d_{\tilde{r}} \left(\tilde{r}_1,\tilde{r}_2\right) + d_{\varpi} \left(\left(\varpi_N\right)_1, \left(\varpi_N\right)_2\right) + d_{\psi} \left(\psi_1,\psi_2\right) \, .
\]
We denote by $\mathcal{B}_b$ the ball of radius $b$ centred around $\left(\frac{u-v}{2}, M_N,\Psi \rho^{\frac{3}{2}-\kappa}\right)$ where we recall $\Psi = \lim_{(u,v_0) \rightarrow (u_0,v_0)} \psi \tilde{r}^{-\frac{3}{2}}$ from Definition \ref{def:id} and $M_N$ from (\ref{bco}).

\subsection{The contraction map\label{contractionsec}}

We now define a map $\Phi : \mathcal{B}_b \ni \left(\tilde{r}, \varpi_N, \psi\right) \mapsto \left(\widehat{\tilde{r}}, \widehat{\varpi_N}, \widehat{\psi} \right)$ by
\begin{align} \label{req}
\widehat{\tilde{r}} = \bar{\tilde{r}}\left(u\right) -  \bar{\tilde{r}}\left(v\right) + \int_v^u du^\prime \int_{u_0}^v  dv^\prime \left[\frac{\Omega^2}{r^2} \left(\frac{3 \varpi_N}{2 r^2} -\frac{1}{2r} + \frac{2 \pi  r \psi^2}{l^2} \left(-a+\frac{3}{2}g\right)\right) \right] 
\end{align}
\begin{equation} \label{oldpsi}
\begin{split}
\widehat{\psi} := &\textrm{Unique $\underline{H}^1$ solution of $-\tilde{\nabla}^\dagger_\mu \tilde{\nabla}^\mu \widehat{\psi} - V\left(\psi,\varpi, r\right) \psi = 0$. } \\
& \textrm{with boundary condition }\ 
\rho^{-\frac{1}{2}-\kappa} \left(\tilde{\partial}_{u} - \tilde{\partial}_v\right) \hat{\psi} + 2\beta \left(t\right) \rho^{-\frac{3}{2}+\kappa} \hat{\psi}=\gamma\left(t\right)
\end{split}
\end{equation}
\begin{equation}
\begin{split} \label{conm}
\widehat{\varpi}_N = \overline{\varpi_N}\left(u\right) + \int_{v_0}^v dv^\prime \Big[ -8\pi r^2 \frac{{r}_u}{{\Omega}^2}  \left[ f \partial_v \left(\frac{\widehat{\psi}}{f}\right) \right]^2+ 4\pi g \left(r-2\varpi_N\right) \widehat{\psi} \left( f \partial_v \frac{\widehat{\psi}}{f} \right) \\
+  2\pi \widehat{\psi}^2  r_v \left( g^2 \left(1-\frac{2\varpi_N}{r}\right) \right) -16\pi^2g^2 \frac{r^3}{l^2} \widehat{\psi}^3 \left( f \partial_v \frac{\widehat{\psi}}{f} \right) - 8\pi^2 g^3 \frac{r^2}{l^2} r_v \widehat{\psi}^4 \Big] \left(u,v^\prime\right)
\end{split}
\end{equation}

\begin{Remark}
See Proposition \ref{wpprop} for the well-posedness of (\ref{oldpsi}).

The fact that $\widehat{\psi}$ (and not $\psi$ itself) appears on the right hand side of (\ref{conm}) is merely technical as we will show existence of a fixed point. It somehow reflects the fact the true dynamics is in the gauge function $\tilde{r}$ and the free field $\psi$. To the same effect, we could have moreover replaced $\tilde{r}$ by $\widehat{\tilde{r}}$ in (\ref{conm}) but prefer not to.
\end{Remark}

We now state the main technical results of this section, and indeed of the paper, as two propositions:
\begin{Proposition} \label{prop:maptoball}
The map $\Phi$ is well-defined and for sufficiently small $\delta$, $\Phi$ in fact maps the ball $\mathcal{B}_b$ into itself.
\end{Proposition}
Furthermore, we have
\begin{Proposition} \label{prop:conmap}
For $\delta$ sufficiently small, $\Phi:\mathcal{B}_b\to \mathcal{B}_b$ is a contraction with respect to the distance $d$.
\end{Proposition}
From these immediately follows the corollary
\begin{Corollary}\label{maincor}
There exists a unique weak solution $(\rt, \varpi_N, \psi) \in \mathcal{B}_b$  of the equations \eq{e3},  \eq{e2}, \eq{e1} which satisfies the initial and boundary conditions, as in \S \ref{initdatasec}, \ref{sec:bdycondition}.
\end{Corollary}
\begin{proof}
By the Banach fixed point theorem, $\Phi$ has a unique fixed point. By construction of $\Phi$, a point $(\rt, \varpi_N, \psi) \in \mathcal{B}_b$ is a fixed point of $\Phi$ if and only if it solves \eq{e3},  \eq{e2}, \eq{e1}.
\end{proof}

The remainder of this section deals with the proof of Propositions \ref{prop:maptoball}, \ref{prop:conmap}. In \S\ref{sec:auxlemmas} we prove some useful auxilliary lemmas, before proving Proposition \ref{prop:maptoball} in \S\ref{sec:maptoball} and Proposition \ref{prop:conmap} in \S\ref{sec:conmap}.

\subsection{Properties of $\mathcal{B}_b$} \label{sec:auxlemmas}
Before we prove Propositions \ref{prop:maptoball}, \ref{prop:conmap}, we first establish some properties of elements in the ball $\mathcal{B}_b$. We denote by $C_b$ a constant depending only on $b$ (the size of the ball) and possibly the initial data quantities $\Psi$, $M_N$ and the parameter $\kappa$ determined by the Klein-Gordon mass via (\ref{kgmass}).

\begin{Lemma} \label{lem:basicballbounds}
Let $\left(\tilde{r}, \varpi_N, \psi\right) \in \mathcal{B}_b \subset \mathcal{C}$. Then we have the following estimates for $\tilde{r}$:
\begin{align}
e^{-b} \leq \frac{2\tilde{r}}{u-v} \leq e^b \ \ \ \ , \ \ \ \ e^{-b} \leq 2\tilde{r}_u \leq e^b \ \ \ \ , \ \ \ \ e^{-b} \leq -2\tilde{r}_v \leq e^b
\end{align}
\begin{align}
|\tilde{r}_{uv}| \leq b  \ \ \ \ , \ \ \ \ \ | T\left(\tilde{r}\right) | \leq b \cdot \rho
\end{align}
Finally, the auxiliary variables $\Omega$ and $\varpi$ satisfy
\begin{align}
\varpi \leq C_b \cdot \tilde{r}^{-2\kappa} \ \ \ , \ \ \  \Big|\frac{\Omega^2}{r^2}\Big| \leq C_b
\end{align}
and $\psi$ satisfies
\[
|\psi| \leq C_b \cdot \tilde{r}^{\frac{3}{2}-\kappa} \ \ \ \ \ and \ \ \ \ \ |\hat{\partial}_u \psi| \leq C_b \cdot \tilde{r}^{\frac{1}{2}-\frac{s}{4}} 
\]
\end{Lemma}
\begin{proof}
Straightforward computation.
\end{proof}
\begin{Corollary}
The function $\tilde{r}$ extends continuously to the boundary $v=u$. The functions $\tilde{r}_u$ and $\tilde{r}_v$ extend to bounded functions on the boundary.
\end{Corollary}
\begin{Corollary} \label{cor:sc}
In addition to the above bounds we have
\begin{align}
\Big| \frac{\tilde{r}_v}{\tilde{r}} + \frac{1}{2\rho} \Big| \leq 3b \cdot e^{b}  \ \ \ \ \ \ , \ \ \ \ \ \Big| \frac{\tilde{r}_u}{\tilde{r}} - \frac{1}{2\rho} \Big| \leq 3b \cdot e^{b}
\end{align}
\end{Corollary}
\begin{proof}
We start from the following inequality which holds in the triangle $\Delta_{\delta,u_0}$:
\[
\partial_u \left(\rho \tilde{r}_v + \frac{1}{2} \tilde{r}\right) = \frac{1}{2} T\left(\tilde{r}\right) + \rho \cdot \tilde{r}_{uv} \,.
\]
The quantity in brackets on the left extends to zero on the boundary $v=u$. Integrating the right hand side yields 
\[
\int du \left[ \frac{1}{2} T\left(\tilde{r}\right) + \rho \cdot \tilde{r}_{uv} \right] \leq \int du \left(\frac{1}{2}b \cdot \rho + b \cdot \rho\right) \leq 3b\cdot \rho^2 \,.
\]
Dividing the resulting integrated inequality by $\tilde{r} \cdot \rho$ and using the first bound of the Lemma yields the desired (first) inequality. The second is proven analogously. 
\end{proof} 

The following version of the previous Corollary for differences will also be useful in the sequel:
\begin{Corollary} \label{cor:diffss}
Given two elements  $\left(\tilde{r}_1, (\varpi_N)_1, \psi_1\right),  \left(\tilde{r}_2, (\varpi_N)_2, \psi_2\right) \in \mathcal{B}_b \subset \mathcal{C}$ we have the estimate
\begin{equation} 
\begin{split}
\Big| \frac{(\tilde{r}_1)_v}{\tilde{r}_1} - \frac{(\tilde{r}_2)_v}{\tilde{r}_2} \Big|
&\leq e^b \left(1+3b\cdot e^{b}\right) \left[\sup_{\Delta} \Big| \frac{T\left(\tilde{r}_1 - \tilde{r}_2\right)}{\rho} \Big| + \sup_{\Delta} \Big|(\tilde{r}_1)_{uv} - (\tilde{r}_2)_{uv} \Big| +\sup_{\Delta} \Big| \frac{\tilde{r}_1 - \tilde{r}_2}{\rho} \Big| \right] \nonumber \\
&\leq C_b \cdot d_{\tilde{r}}\left(\tilde{r}_1,\tilde{r}_2\right)
\end{split}
\end{equation}
and the same estimate with $v$ replaced by $u$ on the left hand side.
\end{Corollary}
\begin{proof}
Note first that similar to the previous corollary we have
\[
\partial_u \left(\rho \left((\tilde{r}_1)_v - (\tilde{r}_2)_v\right) + \frac{1}{2} \left(\tilde{r}_1 - \tilde{r}_2 \right)\right) = \frac{1}{2} T\left(\tilde{r}_1 - \tilde{r}_2\right) + \rho \cdot \left((\tilde{r}_1)_{uv} - (\tilde{r}_2)_{uv}\right) \, ,
\]
which after integration leads to
\begin{align} \label{im}
|\rho \left((\tilde{r}_1)_v - (\tilde{r}_2)_v\right) + \frac{1}{2} \left(\tilde{r}_1 - \tilde{r}_2\right) | \leq \rho^{2} \left[\sup_{\Delta} \Big| \frac{T\left(\tilde{r}_1 - \tilde{r}_2\right)}{\rho} \Big| + \sup_{\Delta} \Big|(\tilde{r}_1)_{uv} - (\tilde{r}_2)_{uv} \Big| \right] \, .
\end{align}
Secondly, observe that we can write
\begin{align}
\Big| \frac{(\tilde{r}_1)_v}{\tilde{r}_1} - \frac{(\tilde{r}_2)_v}{\tilde{r}_2} \Big|   \leq \frac{1}{\tilde{r}_1} \Big|\left((\tilde{r}_1)_v - (\tilde{r}_2)_v \right) + \frac{(\tilde{r}_2)_v}{\tilde{r}_2} \left(\tilde{r}_2-\tilde{r_1}\right) \Big| \nonumber \\
\leq \rho^{-1} e^b  \left( \Big|\left((\tilde{r}_1)_v - (\tilde{r}_2)_v \right) + \frac{1}{2\rho} \left(\tilde{r}_1-\tilde{r_2}\right) \Big| + \Big| \frac{(\tilde{r}_2)_v}{\tilde{r}_2} + \frac{1}{2\rho} \Big| \cdot |\tilde{r}_2-\tilde{r}_1|\right) \nonumber \\
\leq \rho^{-2} e^b \Big|\rho \left((\tilde{r}_1)_v - (\tilde{r}_2)_v \right) + \frac{1}{2} \left(\tilde{r}_1-\tilde{r_2}\right) \Big| + e^b \cdot 3b e^{b} 
 \sup_{\Delta_\delta} \Big| \frac{\tilde{r}_1 - \tilde{r}_2}{\rho} \Big| \, ,\nonumber
\end{align}
where we have used Corollary \ref{cor:sc} in the last step. Inserting (\ref{im}) yields the result. The $u$-direction is proven analogously.
\end{proof}

Corollary \ref{cor:sc} allows us to establish the equivalence 
between the twisted derivatives defined by (\ref{htd2}) and (\ref{ttd}).
Indeed, the identity
 \[
 \tilde{\partial}_u \psi = \widehat{\partial}_u \psi + \psi \left(\frac{3}{2}-\kappa\right) \left(\frac{1}{2\rho} - \frac{\tilde{r}_u}{\tilde{r}}\right)
 \]
 immediately proves
\begin{Lemma} \label{lem:twisteq}
Let $\left(\tilde{r}, \varpi_N, \psi\right) \in \mathcal{B}_b \subset \mathcal{C}$. Then we have
\begin{align}
\frac{1}{C_b} \left[ \left(\hat{\partial}_u \psi\right)^2 + \psi^2 \right] \rho^{-2}  \leq \left[ \left(\tilde{\partial}_u \psi\right)^2 + \psi^2 \right] \tilde{r}^{-2} \leq C_b \left[ \left(\hat{\partial}_u \psi\right)^2 + \psi^2 \right] \rho^{-2} \, .
\end{align}
\end{Lemma}

This Lemma will be useful, because the energy estimates will turn out to naturally twist with $\tilde{r}$. The Lemma guarantees that for the norm (\ref{c0h1}) twisting with $\rho$ or $\tilde{r}$ are equivalent.

\subsection{Map to the ball (Proof of Proposition \ref{prop:maptoball})} \label{sec:maptoball}

\subsubsection*{The radial bounds}

We first verify that the contraction map respects the boundary conditions required of $\rt$. To do so, note that the integrand in (\ref{req}) satisfies
\[
\Big| \left[ \frac{\Omega^2}{r^2} \left(...\right) \right] \Big| \leq C_b \cdot \tilde{r}^{\min(1,2-2\kappa)}
\]
and is hence integrable in $v$. Therefore $\widehat{\tilde{r}}|_{\scri} = 0$ on the boundary. Note also that $T \left(\widehat{\tilde{r}}\right) :=\left(\partial_u + \partial_v\right)\widehat{\tilde{r}}$ extends to zero on the boundary by the dominant convergence theorem. Moreover, clearly $\widehat{\tilde{r}}\left(u,u_0\right) = \overline{\tilde{r}}\left(u\right)$.

We now compute
\begin{align}
\widehat{\tilde{r}}_{uu} \left(u,v\right)  =  \overline{\tilde{r}}_{uu} \left(u\right) +  \int_{u_0}^v  dv^\prime \partial_u \left[\frac{\Omega^2}{r^2} \left(\frac{3 \varpi_N}{2 r^2} -\frac{1}{2r} + \frac{2 \pi  r \left(f \psi\right)^2}{f^2 l^2} \left(-a+\frac{3}{2}g\right)\right) \right] \left(u,v^\prime\right) \nonumber \, .
\end{align}
Writing $\frac{\Omega^2}{r^2} = - \frac{4\tilde{r}_u \tilde{r}_v}{\left(1-\mu\right) \tilde{r}^2}$ and using the properties of the element of the ball it is not hard to see that the integrand can be bounded pointwise by
\begin{align}
| \partial_u \left[... \right] | \leq C_b \cdot \tilde{r}^{\min\left(0,1-2\kappa\right)} \, ,
\end{align}
which is integrable for $0<\kappa<1$ and hence
\begin{align}
| \widehat{\tilde{r}}_{uu} \left(u,v\right)  -  \overline{\tilde{r}}_{uu} \left(u\right)| \leq \frac{C_b}{\min\left(1,2-2\kappa\right)} \delta^{\min\left(1,2-2\kappa\right)} \, ,
\end{align}
which means that for sufficiently small $\delta$
\[
| \widehat{\tilde{r}}_{uu} \left(u,v\right) | \leq \delta^\prime + \frac{b}{100} \, .
\]
Similarly,
\begin{align}
|\widehat{\tilde{r}}_{uv} \left(u,v\right)| \leq C_b \cdot \tilde{r}^{min\left(1,2-2\kappa\right)} \leq C_b \cdot \delta^{min\left(1,2-2\kappa\right)}
\end{align}
leads to 
\[
| \widehat{\tilde{r}}_{uv} \left(u,v\right) | \leq  \frac{b}{100} \, .
\]
The lower derivatives are also straightforward:
\begin{align}
| \widehat{\tilde{r}}_u \left(u,v\right) - \overline{\tilde{r}}_u \left(u\right) | \leq C_b \cdot \delta^{min\left(2,3-2\kappa\right)} \, ,
\end{align}
\begin{align}
| \widehat{\tilde{r}}_v \left(u,v\right) + \overline{\tilde{r}}_u \left(v\right) | \leq C_b \cdot \delta^{min\left(2,3-2\kappa\right)} \, ,
\end{align}
\begin{equation}
\begin{split}
\widehat{\tilde{r}} \leq \sup \overline{\tilde{r}}_u \left(u-v\right)  + C_b \cdot  \delta^{min \left(2,3-2\kappa\right)} \left(u-v\right) < \left(\frac{1}{2} + \delta^\prime + C_b \delta\right) \rho \, , \\
\widehat{\tilde{r}} \geq \inf \overline{\tilde{r}}_u \left(u-v\right)  - C_b \cdot  \delta^{min\left(2,3-2\kappa\right)} \left(u-v\right) \geq \left(\frac{1}{2} - \delta^\prime - C_b \delta\right) \rho \, ,
\end{split}
\end{equation}
which implies that for $\delta$ sufficiently small,
\[
\Big| \log \frac{\widehat{\tilde{r}}}{\rho} \Big| + | \log 2\widehat{\tilde{r}}_u |+ | \log \left(-2\widehat{\tilde{r}}_v\right) | < \frac{b}{100} \, .
\]
 Finally, note that indeed $T\left(\widehat{\tilde{r}}\right)= \left( \partial_u + \partial_v\right)\widehat{\tilde{r}}$ vanishes on the boundary $u=v$ and hence
\begin{align}
T\left(\widehat{\tilde{r}}\right) \left(u,v\right) &= 0 + \int_v^u du^\prime \left( \widehat{\tilde{r}}_{vu} +    \widehat{\tilde{r}}_{uu} \right) \left(u^\prime,v\right) \, ,
\end{align}
which using that the integrand is $\delta$ small by previous bounds leads to
\[
\Bigg| \frac{T\left(\widehat{\tilde{r}}\right)}{\rho} \Bigg| \leq C_b \cdot \delta^{min\left(1,2-2\kappa\right)} < \frac{b}{100} \, .
\]
In summary, for sufficiently small $\delta$ we indeed map back into the ball.

\subsubsection*{Estimates for $\psi$}
From the wave equation we derive 
\begin{align}
\frac{1}{2}\partial_u \left(f^2 r^2 \left(\partial_v \frac{\widehat{\psi}}{f}\right)^2 + r^2 \hat{\psi}^2\right) + \frac{1}{2}\partial_v \left(f^2 r^2 \left(\partial_u \frac{\widehat{\psi}}{f}\right)^2 + r^2 \hat{\psi}^2 \right) \nonumber \\
= -\frac{T\left(rf\right)}{rf} \cdot f r \left(\partial_v \frac{\widehat{\psi}}{f}\right) \cdot f r \left(\partial_u \frac{\widehat{\psi}}{f}\right) + \frac{T(r)}{r} r^2 \hat{\psi}^2 \nonumber \\
+\frac{\Omega^2}{4r^2} \cdot r^3 V\left(\psi,\varpi,r\right) \psi \cdot f r \left( \partial_v \frac{\widehat{\psi}}{f} + \partial_u \frac{\widehat{\psi}}{f}  \right) + r^2 \hat{\psi} f \left( \partial_v \frac{\widehat{\psi}}{f} + \partial_u \frac{\widehat{\psi}}{f} +\frac{T(f)}{f^2} \hat{\psi} \right) \nonumber
\end{align}
Integrating this over space-time and using that
\begin{align}
\Big|\frac{T\left(rf\right)}{rf} \Big| + \Big|\frac{T\left(r\right)}{r} \Big| \leq C_b
\end{align}
holds for elements in the ball, we can estimate the first two terms on the right hand side by
\begin{align}
\int_{\Delta} \text{second line} \leq C_b \cdot \delta \cdot \|\widehat{\psi}\|^2_{C^0 \underline{H}^1\left(\Delta\right)}
\end{align}
where we recall Lemma \ref{lem:twisteq} (ensuring the equivalence between twisting with $\rho$ and $\tilde{r}$ as far as the $H^1$-norm is concerned) and the third line by
\begin{align}
\int_{\Delta} \text{third line} \leq C_b \cdot \delta \cdot \|\widehat{\psi}\|^2_{C^0 \underline{H}^1\left(\Delta\right)} + \int_{\Delta} \left(\frac{\Omega^2}{r^2}\right)^2 V^2 r^6 \psi^2
\end{align}
and naively applying pointwise bounds
\begin{align} \label{enres}
\int_{\Delta} \left(\frac{\Omega^2}{r^2}\right)^2 V^2 r^6 \psi^2 \leq C_b \int_{\Delta} \rho^{2\cdot min(2,3-2\kappa)} \rho^{-6} \rho^{3-2\kappa} du dv \leq C_b \cdot \delta
\end{align}
The last step follows from our assumption $\kappa<\frac{2}{3}$, which implies $3-6\kappa>-1$ and makes the expression integrable. The $\delta$ is coming from the integration in the other direction.
To compute the boundary term on $\scri$  we recall the boundary 
condition
\begin{align} \label{byco}
\rho^{-\frac{1}{2}-\kappa} \left(\tilde{\partial}_{u} - \tilde{\partial}_v\right) \hat{\psi} + 2\beta \left(t\right) \rho^{-\frac{3}{2}+\kappa} \hat{\psi}=\gamma\left(t\right) \, .
\end{align}
Hence the boundary term on $\scri$ becomes\footnote{Notice that it suffices to assume that $\gamma \in H^{1-\kappa}$ to deal with the inhomogeneous term, which will be relevant when we later consider non-linear boundary conditions.}
\begin{align}
\int_{\scri} dt \frac{1}{2} r^2 f^2 \left(\partial_v \frac{\widehat{\psi}}{f}+\partial_u \frac{\widehat{\psi}}{f}\right) \left(\partial_v \frac{\widehat{\psi}}{f}-\partial_u \frac{\widehat{\psi}}{f}\right) \nonumber \\
= \int_{\scri} dt \frac{1}{2} r^2 \rho^{-1+2\kappa} (-\beta) f^2 \left(\partial_u + \partial_v \right) \left[ \frac{\widehat{\psi}}{f}\right]^2 + \int_{\scri} dt \frac{1}{2} r^2 f \left(\partial_v \frac{\widehat{\psi}}{f}+\partial_u \frac{\widehat{\psi}}{f}\right) \gamma\left(t\right) \rho^{\frac{1}{2}+\kappa}\nonumber \\
= \frac{1}{2} \beta \left(t=u_0\right) \Psi^2 - \frac{1}{2}\beta \left(t\right)\widehat{\psi}^2 \tilde{r}^{-2} \rho^{-1+2\kappa} \Big|_{\left(u_0+\delta,u_0+\delta\right)} + \int_{\scri} dt \frac{1}{2} T\left(  \tilde{r}^{-2} \rho^{-1+2\kappa}\beta f^2\right)  \left[ \frac{\widehat{\psi}}{f}\right]^2 \nonumber \\
-  \frac{1}{2}\Psi  \gamma \left(t=u_0\right)  + \frac{1}{2}\gamma \left(t\right)\cdot \widehat{\psi}  \tilde{r}^{-2} \rho^{1/2+\kappa} \Big|_{\left(u_0+\delta,u_0+\delta\right)} - \int_{\scri} dt \frac{1}{2} T\left(  \tilde{r}^{-2} \rho^{1/2+\kappa} \gamma\left(t \right)\right)  \hat{\psi} \, .
\nonumber
\end{align}
Since $T\left(\tilde{r}\right) \sim \tilde{r}$, the last term in both the third and the fourth line are easily seen to be controlled by $\delta \cdot \|\widehat{\psi}\|_{C^0\underline{H}^1}^2$ and $\delta \cdot \|\widehat{\psi}\|_{C^0\underline{H}^1}$ respectively, the $\delta$ coming from integration in $t$. Therefore,
\begin{align} 
\Big| \int_{\scri} dt \left( ... \right) \Big| \leq \frac{1}{2}|\beta\left(u_0+\delta\right)| \cdot
\sup_{\Delta_{\delta,u_0}} | \tilde{r}^{-2} \rho^{-1+2\kappa} \psi^2 - \Psi^2 |  + \frac{1}{2}\Psi^2  |\beta \left(u_0+\delta\right) -\beta\left(u_0\right)| \nonumber \\
+ \frac{1}{2} |\gamma\left(u_0+\delta\right)|\cdot \sup_{\Delta_{\delta,u_0}} |  \tilde{r}^{-2} \rho^{1/2+\kappa} \psi - \Psi|
+ \frac{1}{2} |\Psi| |\gamma\left(u_0+\delta\right) - \gamma\left(u_0\right)| +C_b \delta \cdot \|\widehat{\psi}\|_{C^0\underline{H}^1}^2 \, , \nonumber
 \end{align}
 of which the second term in both lines can be estimated by $C_{\Psi,\beta,g} \cdot \delta$, provided that $\beta$ and $g$ are $C^1$. For the terms involving the sup, we recall that $|\frac{\tilde{r}}{\rho}-1|$ is $\delta$-small in $\Delta_{\delta,u_0}$ (integrate $\partial_v \left(\frac{\tilde{r}}{\rho}-1\right)$ which is uniformly bounded by Corollary \ref{cor:sc} from initial data where it is $\delta^\prime$-small) and that we also have (cf.~(\ref{sty})),
\begin{align}
| \rho^{-\frac{3}{2}+\kappa} \widehat{\psi} (u,v) - \Psi | &\leq | \rho^{-\frac{3}{2}+\kappa} \widehat{\psi}(u,v) - \overline{\rho}^{-\frac{3}{2}+\kappa} \overline{\psi}(u) | +  | \overline{\rho}^{-\frac{3}{2}+\kappa} \overline{\psi} (u) -\Psi |   \nonumber \\
&\leq C_b \cdot \delta^\kappa \cdot \|\widehat{\psi}\|_{C^0 \underline{H}^1\left(\Delta\right)}  \, . \nonumber
\end{align}
Combining all of the above, for sufficiently small $\delta$ we obtain 
\begin{align} \label{fo1}
\|\widehat{\psi}\|^2_{C^0 \underline{H}^1\left(\Delta\right)} &\leq 4 \int_{u_0}^{u_0+\delta} dur^2 \left[ f^2  \left(\partial_u \frac{\widehat{\psi}}{f}\right)^2 + \widehat{\psi}^2\right] + C_{b,\beta,g,\Psi} \cdot \delta + C_b \cdot \delta^\kappa \cdot \|\widehat{\psi}\|_{C^0 \underline{H}^1\left(\Delta\right)} \nonumber \, .
\end{align}
Applying Cauchy's inequality to the last term we find
\begin{equation} \label{gnu}
\|\widehat{\psi}\|_{C^0 \underline{H}^1\left(\Delta\right)}  \leq 3 \left(\delta^\prime\right)^\frac{1}{2} + C_{b,\beta,g,\Psi} \cdot \delta^\kappa <  \frac{b}{100} \, 
\end{equation}
and of course also immediately
\[
| \rho^{\frac{3}{2}-\kappa} \widehat{\psi} (u,v) - \Psi | <  C_b \cdot \delta^\kappa \cdot \|\widehat{\psi}\|_{C^0 \underline{H}^1\left(\Delta\right)}  < \frac{b}{100} \, .
\]
Writing the wave equation as a transport equation, we can also retrieve the pointwise bound for the $u$-derivative: Starting from
\begin{align}
\partial_v \left(f r \left(\partial_u \frac{\widehat{\psi}}{f}\right)\right) = - \partial_u \left(rf\right) \left(\partial_v \frac{\widehat{\psi}}{f}\right) - \frac{\Omega^2}{4} r V \hat{\psi} 
\end{align}
we derive
\begin{align}
\Big | f r \left(\partial_u \frac{\widehat{\psi}}{f}\right) \left(u,v\right) \Big| &\leq \Big | f r \left(\partial_u \frac{\widehat{\psi}}{f}\right) \left(u,v_0\right) \Big| + C_b \cdot r^\frac{1}{2} \cdot \|\widehat{\psi}\|_{C^0 \underline{H}^1\left(\Delta\right)} + C_b \int_{v_0}^v r^{3-\frac{3}{2}+\kappa} V \left(|\Psi| + \frac{b}{100}\right) \nonumber \\
&\leq C_b \cdot C_{data} \cdot r^{\frac{1}{2}} \left(u,v_0\right) +  C_b \cdot \delta^\kappa \cdot r^\frac{1}{2} +C_b \left(|\Psi| + 1\right) r^{\frac{3}{2}+\kappa-\min\left(2,3-2\kappa\right)-1}
\end{align}
and since both $-\frac{5}{2}+3\kappa < \frac{1}{2}$ and $-\frac{3}{2}+\kappa< \frac{1}{2}$ holds for $0<\kappa<1$ we obtain after using that $r\left(u,v\right)>r\left(u,v_0\right)$ in $\Delta_{\delta,u_0}$, the bound
\begin{align}
\Big | f r^\frac{1}{2} \left(\partial_u \frac{\widehat{\psi}}{f}\right) \left(u,v\right) \Big| \leq C_{b} \cdot C_{data} + C_{b,\Psi} \cdot \delta^{\min(3-3\kappa, 2-\kappa)} \, .
\end{align}
Clearly
\[
 f r^\frac{1}{2} \left(\partial_u \frac{\widehat{\psi}}{f}\right) =  r^{\frac{1}{2}} \left(\hat{\partial}_u \widehat{\psi} \right) + r^{\frac{1}{2}} \widehat{\psi}\left[\frac{\rho}{\tilde{r}}\right]^{-\frac{3}{2}+\kappa} \partial_u \left(\left[\frac{\rho}{\tilde{r}}\right]^{\frac{3}{2}-\kappa}\right)
\]
and since Corollary \ref{cor:sc} establishes boundedness for the $u$-derivative of the round bracket in the last term, we finally obtain
\[
\Big| \rho^{-\frac{1}{2}+\frac{s}{4}} \left(\hat{\partial}_u \widehat{\psi} \right) \Big| \leq \left(C_b \cdot C_{data} + C_{b,\Psi} \cdot \delta^{\min(3-3\kappa, 2-\kappa)}\right) \delta^\frac{s}{4} < \frac{b}{100}\, ,
\]
for sufficiently small $\delta$ depending only on the initial data constant $\Psi$ and $b$. 

\subsubsection*{The renormalized mass}
Using Cauchy-Schwarz and basic properties of the weights following from Lemma \ref{lem:basicballbounds} it is not hard to see that for $\delta$ sufficiently small
\begin{align} 
|\widehat{\varpi}_N - M_N| \leq \delta^\prime + C_b \cdot \|\widehat{\psi}\|^2_{C^0 \underline{H}^1\left(\Delta\right)} < \frac{b}{100}\, .
\end{align}
Taking a $u$-derivative one establishes after a tedious computation\footnote{Here we only mention the most critical term arising in this computation which is estimated
\begin{align}
\int_{v_0}^v dv^\prime \Big | 8\pi r^2 \frac{r_u}{\Omega^2} \frac{\partial_v \left(f r\right)}{fr} \tilde{\partial}_v \hat{\psi} \tilde{\partial}_u \hat{\psi}\Big| \leq  C_b \|\hat{\psi}\|_{C^0 \underline{H}^1\left(\Delta\right)} \sqrt{ \int_{v_0}^v r^{3+\frac{s}{2}} dv^\prime } \leq C_b \cdot \delta \cdot \rho^{-1-\frac{s}{4}} \left(u,v\right)
\end{align}
providing the required smallness for (\ref{tgt}) after multiplication by $\rho^{1+s}\left(u,v\right)$.
} 
using the wave equation that
\begin{align} \label{tgt}
|\partial_u\widehat{\varpi}_N \left(u,v\right) | \leq \delta^\prime \tilde{r}^{-1-s} \left(u,v_0\right) + C_{b,\Psi} \tilde{r}^{-1-\frac{s}{4}} \left(u,v\right) \|\widehat{\psi}\|^2_{C^0 \underline{H}^1\left(\Delta\right)}  \, ,
\end{align}
which after multiplying by $\rho^{1+s}$ retrieves also
\[
|\rho^{1+s} \partial_u\widehat{\varpi}_N| < \frac{b}{100} \, .
\]

\subsection{Contraction property  (Proof of Proposition \ref{prop:conmap})}\label{sec:conmap}
Let $\left(\widehat{\tilde{r}}_1,\left(\widehat{\varpi}_N\right)_1,\widehat{\psi}_1\right)$ and  $\left(\widehat{\tilde{r}}_2, \left(\widehat{\varpi}_N\right)_2,\widehat{\psi}_2\right)$ be two points in $\mathcal{B}$. To establish the contraction property it suffices to prove
\begin{align} \label{guh}
d\left( \left(\widehat{\tilde{r}}_1, \left(\widehat{\varpi}_N\right)_1,\widehat{\psi}_1\right), \left(\widehat{\tilde{r}}_2,\left(\widehat{\varpi}_N\right)_2,\widehat{\psi}_2 \right)\right)  \leq \frac{1}{2}\cdot  d\left( \left({\tilde{r}}_1, \left({\varpi}_N\right)_1,{\psi}_1\right), \left(\tilde{r}_2,\left({\varpi}_N\right)_2,{\psi}_2 \right) \right)  \, .
\end{align}
We begin with a few decomposition formulae:
\begin{Lemma} \label{lem:cribo}
We have
\begin{align}
|r_2 - r_1|  \leq \frac{C_b}{\rho} \cdot d_{\tilde{r}}\left(\tilde{r}_1,\tilde{r}_2\right) \, ,
\end{align}
\begin{align}  \label{omdeco}
\left(\Omega_2\right)^2 -\left(\Omega_1\right)^2 \leq  \frac{C_b}{\rho}  \cdot  d\left( \left({\tilde{r}}_1, \left({\varpi}_N\right)_1,{\psi}_1\right), \left(\tilde{r}_2,\left({\varpi}_N\right)_2,{\psi}_2 \right) \right)  \, ,
\end{align}
\begin{align} \label{omdeco2}
\Big | \frac{(r_1)_u}{(\Omega_1)^2} - \frac{(r_2)_u}{(\Omega_2)^2} \Big| \leq C_b \cdot  d\left( \left({\tilde{r}}_1, \left({\varpi}_N\right)_1,{\psi}_1\right), \left(\tilde{r}_2,\left({\varpi}_N\right)_2,{\psi}_2 \right) \right)  \, .
\end{align}
\end{Lemma}
\begin{proof}
This follows from the computations:
\begin{align}
|r_2 - r_1| = \Big|\frac{1}{(\tilde{r}_2)}-\frac{1}{(\tilde{r}_1)}\Big| = \Big| \frac{1}{\tilde{r}_1 \tilde{r}_2} \left(\tilde{r}_1-\tilde{r}_2\right)\Big| \leq \frac{C_b}{\rho} \sup \Big | \frac{\tilde{r}_1-\tilde{r}_2}{\rho} \Big| 
\end{align}
and
\begin{align}
\left(\Omega_2\right)^2 -\left(\Omega_1\right)^2 = -4\frac{(\tilde{r}_2)_u}{\tilde{r}_2} \frac{(\tilde{r}_2)_v}{\tilde{r}_2} \frac{(r_2)^2}{1-\mu_2} + 4\frac{(\tilde{r}_1)_u}{\tilde{r}_1} \frac{(\tilde{r}_1)_v}{\tilde{r}_1} \frac{(r_1)^2}{1-\mu_1} \nonumber \\
= -4 \left( \frac{(\tilde{r}_2)_u}{\tilde{r}_2} - \frac{(\tilde{r}_1)_u}{\tilde{r}_1} \right) \frac{(\tilde{r}_2)_v}{\tilde{r}_2} \frac{(r_2)^2}{1-\mu_2} - 4 \frac{(\tilde{r}_1)_u}{\tilde{r}_1} \left( \frac{(\tilde{r}_2)_v}{\tilde{r}_2} - \frac{(\tilde{r}_1)_v}{\tilde{r}_1} \right) \frac{(r_2)^2}{1-\mu_2} \nonumber \\
-4 \frac{(\tilde{r}_1)_u}{\tilde{r}_1} \frac{(\tilde{r}_1)_u}{\tilde{r}_1} \left(\frac{(r_2)^2}{1-\mu_2} - \frac{(r_1)^2}{1-\mu_1} \right) \nonumber .
\end{align}
Indeed, by Corollary \ref{cor:diffss} and Lemma \ref{lem:basicballbounds} and exploiting a cancellation of the top order term in the last line, namely 
\[
\frac{(r_2)^2}{1-\mu_2} - \frac{(r_1)^2}{1-\mu_1}= \frac{  (1-\frac{2(\varpi_N)_1}{r_1}) (r_2)^2 - (1-\frac{2(\varpi_N)_2}{r_2}) (r_1)^2 -4\pi g \frac{(r_1)^2(r_2)^2}{l^2} \left((\psi_1)^2-(\psi_2)^2\right)}{(1-\mu_1)(1-\mu_2)} \, ,
\]
we obtain after further massaging (\ref{omdeco}). The estimate (\ref{omdeco2}) is then straightforward.
\end{proof}

Turning to the proof of (\ref{guh}), we first establish

\subsubsection*{The radial bound}
\[
d_{\tilde{r}} \left(\tilde{r}_1,\tilde{r}_2\right) \leq C_b \cdot \delta^{\min(1,2-2\kappa)} \cdot  d\left( \left({\tilde{r}}_1, \left({\varpi}_N\right)_1,{\psi}_1\right), \left(\tilde{r}_2,\left({\varpi}_N\right)_2,{\psi}_2 \right) \right)  \, ,
\]
which follows by decomposing
\begin{align} \label{decoi}
\frac{(\Omega_1)^2}{(r_1)^2} \left(\frac{3 (\varpi_N)_1}{2 (r_1)^2} -\frac{1}{2r_1} + \frac{2 \pi  r_1 (\psi_1)^2}{l^2} \left(-a+\frac{3}{2}g\right)\right) - \textrm{same with $1 \leftrightarrow 2$} = ...
\end{align}
as differences, of which we only discuss the most difficult term:
\begin{align} \label{mdte}
\Big| \frac{(\Omega_2)^2}{(r_2)} (\psi_2)^2 - \frac{(\Omega_1)^2}{(r_1)} (\psi_1)^2 \Big| = \Big| \frac{(\psi_2)^2}{r_2} \Big| | \left(\Omega_2\right)^2 -\left(\Omega_1\right)^2 | +\frac{(\Omega_1)^2}{r_2}  |\psi_2 + \psi_1| |\psi_2 -\psi_1| \nonumber \\
+ (\Omega_1)^2 (\psi_1)^2 |\tilde{r}_2 - \tilde{r}_1| \leq \rho^{\min\left(1,2-2\kappa\right)}  d\left( \left({\tilde{r}}_1, \left({\varpi}_N\right)_1,{\psi}_1\right), \left(\tilde{r}_2,\left({\varpi}_N\right)_2,{\psi}_2 \right) \right)  \, ,
\end{align}
which follows by inserting previous bounds on elements in the ball. This already etablishes
\begin{align}
| \hat{\tilde{r}}_1 - \hat{\tilde{r}}_2 | &\leq C_b \int_u^v du^\prime \int_{u_0}^v dv^\prime \rho^{\min\left(1,2-2\kappa\right)} d\left( \left({\tilde{r}}_1, {\varpi}_1,{\psi}_1\right), \left(\tilde{r}_2,{\varpi}_2,{\psi}_2 \right)\right) \nonumber \\
& \leq C_b \cdot \rho \cdot \delta \cdot \delta^{\min\left(1,2-2\kappa\right)} \cdot  d\left( \left({\tilde{r}}_1, \left({\varpi}_N\right)_1,{\psi}_1\right), \left(\tilde{r}_2,\left({\varpi}_N\right)_2,{\psi}_2 \right) \right) \, ,
\end{align}
and similarly
\begin{align}
| (\hat{\tilde{r}}_1)_u - (\hat{\tilde{r}}_2)_u | + | (\hat{\tilde{r}}_1)_v - (\hat{\tilde{r}}_2)_v | \nonumber \\
\leq C_b \cdot \delta \cdot \delta^{\min\left(1,2-2\kappa\right)} \cdot  d\left( \left({\tilde{r}}_1, \left({\varpi}_N\right)_1,{\psi}_1\right), \left(\tilde{r}_2,\left({\varpi}_N\right)_2,{\psi}_2 \right) \right) \, ,
\end{align}
as well as
\begin{align}
| (\hat{\tilde{r}}_1)_{uv} - (\hat{\tilde{r}}_2)_{uv} | \leq C_b \cdot \delta^{\min\left(1,2-2\kappa\right)} \cdot  d\left( \left({\tilde{r}}_1, \left({\varpi}_N\right)_1,{\psi}_1\right), \left(\tilde{r}_2,\left({\varpi}_N\right)_2,{\psi}_2 \right) \right)  \, .
\end{align}
To estimate the differences of $\widehat{\tilde{r}}_{uu}$, we need to differentiate the difference of the integrands (\ref{decoi}) in $u$. Schematically:
\[
(\widehat{\tilde{r}}_1)_{uu} - (\widehat{\tilde{r}}_2)_{uu} = \int_{v_0}^v dv^\prime \left[ \partial_u \left(\textrm{integrand}_1 - \textrm{integrand}_2\right) \right] \, ,
\]
and we need to estimate the integrand analogous to what we did in (\ref{mdte}) for the most difficult term. 
Again we omit this tedious computation and present only the most difficult term:
\begin{align}
\Big| \partial_u \left[ \frac{(\Omega_2)^2}{(r_2)} (\psi_2)^2 \right] - \partial_u \left[ \frac{(\Omega_1)^2}{(r_1)} (\psi_1)^2 \right] \Big| \nonumber \\= \Big| \partial_u \left[ 4(r_2)^3 \frac{-(\tilde{r}_2)_u (\tilde{r}_2)_v}{(1-\mu_2)} (\psi_2)^2 \right] - \partial_u \left[ 4({r}_1)^3 \frac{-(\tilde{r}_1)_u (\tilde{r}_1)_v}{(1-\mu_1)} (\psi_1)^2 \right] \Big| \, , \nonumber
\end{align}
from which we see (counting weights) that the $\tilde{r}_{uu}$ difference (and the $\tilde{r}_{uv}$ difference) will enter with a factor $\rho^{-3+2+3-2\kappa}$. When the derivative hits the $r$- or the $(1-\mu)$-terms we lose one power compared with the computation (\ref{mdte}) and hence obtain $\rho^{\min\left(0,1-2\kappa\right)}$ as a factor, which is integrable for $0<\kappa<1$ and provides the required smallness factor. Finally, when the derivative hits the $(\psi)^2$-term we twist the derivative to obtain a zeroth order term (which loses one power and is hence handled as previously) and
\begin{align}
\Big| \left[ \frac{(\Omega_2)^2}{(r_2)} (\psi_2)(\tilde{\partial}_u \psi_2) \right] - \left[ \frac{(\Omega_1)^2}{(r_1)} (\psi_1)(\tilde{\partial}_u \psi_1) \right] \Big| \nonumber \\
\leq C_b \cdot \rho^{-1+\frac{1}{2}-\frac{s}{4}+\frac{3}{2}-\kappa} \cdot |\rho^{-\frac{1}{2}+\frac{s}{4}} \left( \tilde{\partial}_u \psi_2 -  \tilde{\partial}_u \psi_1\right)| + C_b \cdot \rho^{\frac{3}{2}-\kappa+\frac{1}{2}-\frac{s}{4}-1} \cdot  |\rho^{-\frac{3}{2}+\kappa} \left( \psi_2 - \psi_1\right)| \nonumber \\
+ C_b\cdot \rho^{-2+\frac{3}{2}-\kappa+\frac{1}{2}-\frac{s}{4}} \cdot |\tilde{r}_2 - \tilde{r}_1| + C_b\cdot \rho^{-2+\frac{3}{2}-\kappa+\frac{1}{2}-\frac{s}{4}} \cdot  d\left( \left({\tilde{r}}_1, \left({\varpi}_N\right)_1,{\psi}_1\right), \left(\tilde{r}_2,\left({\varpi}_N\right)_2,{\psi}_2 \right) \right)  \nonumber \, ,
\end{align}
where (\ref{omdeco}) was used. We note that also here all $\rho$-weights are integrable. We conclude
\[
\Big| (\widehat{\tilde{r}}_1)_{uu} - (\widehat{\tilde{r}}_2)_{uu} \Big| \leq \frac{C_b}{2-2\kappa} \cdot \delta^{\min(1,2-2\kappa)} \cdot d\left( \left({\tilde{r}}_1, \left({\varpi}_N\right)_1,{\psi}_1\right), \left(\tilde{r}_2,\left({\varpi}_N\right)_2,{\psi}_2 \right) \right)  \, .
\]
Finally to get the $T\left(\widehat{\tilde{r}}_1- \widehat{\tilde{r}}_2\right)$ difference, we recall that this quantity vanishes on the boundary and therefore integrating from the boundary yields
\[
T\left(\widehat{\tilde{r}}_1- \widehat{\tilde{r}}_2\right) \left(u,v\right) = \int_u^v du^\prime \partial_u \left(\partial_u \left(\widehat{\tilde{r}}_1-\widehat{\tilde{r}}_2\right) + \partial_v \left(\widehat{\tilde{r}}_1-\widehat{\tilde{r}}_2\right) \right) \left(u^\prime,v\right) \, ,
\]
from which we obtain
\[
\Bigg| \frac{T\left(\widehat{\tilde{r}}_1- \widehat{\tilde{r}}_2\right)}{\rho}\Bigg|
 \leq C_b \cdot \delta^{\min(1,2-2\kappa)} \cdot  d\left( \left({\tilde{r}}_1, \left({\varpi}_N\right)_1,{\psi}_1\right), \left(\tilde{r}_2,\left({\varpi}_N\right)_2,{\psi}_2 \right) \right) 
\]
from our previous estimates for the $\widehat{\tilde{r}}_{uu}$ and the $\widehat{\tilde{r}}_{uv}$ difference.

\subsubsection*{Estimates for $\psi$}
The goal is to establish
\begin{align} \label{goaly2}
 d_{\psi} \left(\widehat{\psi}_1,\widehat{\psi}_2\right) \leq C_b \cdot \delta \cdot    d\left( \left({\tilde{r}}_1, \left({\varpi}_N\right)_1,{\psi}_1\right), \left(\tilde{r}_2,\left({\varpi}_N\right)_2,{\psi}_2 \right) \right)  \, .
\end{align}
From the wave equation we derive
\begin{align} \label{wed1}
\partial_u \left(f_1 r_1 \partial_v \left(\frac{\widehat{\psi}_1}{f_1} - \frac{\widehat{\psi}_2}{f_2}\right) \right)= - \partial_v \left(r_1 f_1\right) 
\partial_u \left(\frac{\widehat{\psi}_1}{f_1} - \frac{\widehat{\psi}_2}{f_2}\right)  + \mathcal{E}
\end{align}
and similarly
\begin{align} \label{wed2}
\partial_v \left(f_1 r_1 \partial_u \left(\frac{\widehat{\psi}_1}{f_1} - \frac{\widehat{\psi}_2}{f_2}\right) \right)= - \partial_u \left(r_1 f_1\right) 
\partial_v \left(\frac{\widehat{\psi}_1}{f_1} - \frac{\widehat{\psi}_2}{f_2}\right)  + \mathcal{E}
\end{align}
where
\begin{align}
\mathcal{E}:= \frac{(\Omega_2)^2 r_2 V_2 \psi_2}{4} \left(\frac{r_1 f_1}{r_2 f_2} - 1\right) - f_2 r_2 \partial_v \left(\frac{\widehat{\psi}_2}{f_2} \right)\partial_u \left(\frac{f_1 r_1}{f_2 r_2} \right) \nonumber \\
- f_2 r_2 \partial_u \left(\frac{\widehat{\psi}_2}{f_2} \right)\partial_v \left(\frac{f_1 r_1}{f_2 r_2} \right) - \frac{(\Omega_1)^2}{4} r_1 V_1 \psi_1 + \frac{(\Omega_2)^2}{4} r_2 V_2 \psi_2
\end{align}
is invariant under interchanging $u$ and $v$. 
Now note the identity\footnote{Observe also that the conformally coupled case $\kappa=\frac{1}{2}$ is special.}
\begin{align}
\partial_u \left(\frac{f_1 r_1}{f_2 r_2} \right)= \left(\frac{1}{2}-\kappa\right) \left(\frac{\tilde{r}_1}{\tilde{r}_2}\right)^{\frac{1}{2}-\kappa} \left[ \frac{(\tilde{r}_1)_u}{\tilde{r}_1} - \frac{(\tilde{r}_2)_u}{\tilde{r}_2} \right]
\end{align}
and the same identity replacing $u$ by $v$, which implies
\begin{align} \label{helpi1}
\Big| \partial_u \left(\frac{f_1 r_1}{f_2 r_2} \right) \Big| \leq C_b \cdot  d\left( \left({\tilde{r}}_1, \left({\varpi}_N\right)_1,{\psi}_1\right), \left(\tilde{r}_2,\left({\varpi}_N\right)_2,{\psi}_2 \right) \right)  \, .
\end{align}
On the other hand, one also has by integrating the $v$-version of (\ref{helpi1}) from data (where $r_1=r_2$)
\begin{align} \label{kaux}
\Big| \frac{r_1 f_1}{r_2 f_2} - 1 \Big| \leq  C_b \cdot \delta \cdot  d\left( \left({\tilde{r}}_1, \left({\varpi}_N\right)_1,{\psi}_1\right), \left(\tilde{r}_2,\left({\varpi}_N\right)_2,{\psi}_2 \right) \right) 
\end{align}
and 
\begin{align}
\Big| - \frac{(\Omega_1)^2}{4} r_1 V_1 \psi_1 + \frac{(\Omega_2)^2}{4} r_2 V_2 \psi_2 \Big| \leq C_b \cdot \rho^{\frac{3}{2}-\kappa-3+3-2\kappa} \cdot |\rho^{-\frac{3}{2}+\kappa} \left(\psi_1- \psi_2\right)| \nonumber \\
+ ... \leq  C_b \cdot \rho^{\frac{3}{2}-3\kappa}\cdot d\left( \left({\tilde{r}}_1, {\varpi}_1,{\psi}_1\right), \left(\tilde{r}_2,{\varpi}_2,{\psi}_2 \right) \right) \, .
\end{align}
In the energy estimate we need the \emph{square} of the $\rho$-weight to be integrable which yields $3-6\kappa>-1$ and hence the familiar $\kappa < \frac{2}{3}$. With the above estimates we indeed see
\begin{align}
\int_u^v du^\prime \int_{v_0}^v dv^\prime \mathcal{E}^2 \leq C_b \cdot \delta \cdot \left( d\left( \left({\tilde{r}}_1, \left({\varpi}_N\right)_1,{\psi}_1\right), \left(\tilde{r}_2,\left({\varpi}_N\right)_2,{\psi}_2 \right) \right)  \right)^2 \, ,
\end{align}
the $\delta$ arising from the fact that we integrate in both $u$ and $v$. The only thing missing to close the energy estimate associated with (\ref{wed1}) and (\ref{wed2}) is to estimate the boundary term. A calculation shows that one needs to control
\begin{align} \label{btde}
\int_{\scri} dt \left(...\right) = &-\int \frac{1}{2}(f_1)^2 (\tilde{r}_1)^{-2} \rho^{-1+2\kappa} \beta \cdot T \left(\frac{\widehat{\psi}_1}{f_1} - \frac{\widehat{\psi}_2}{f_2}\right)^2 \, .
\end{align}
Integrating by parts and treating the terms as in the original estimate in Section \ref{sec:maptoball} we control this term by $\delta^\kappa C_b \cdot d\left( \widehat{\psi}_1, \widehat{\psi}_2 \right)$.

In summary, the energy estimate associated with (\ref{wed1}) and (\ref{wed2}) furnishes the estimate
\begin{align} \label{basep1}
 \sup_{(u,v) \in \Delta} \int_v^u  \rho^{-2} \left[ \left(f_1 \partial_u \left( \frac{\widehat{\psi}_1}{f_1}-\frac{\widehat{\psi}_2}{f_2}\right)\right)^2 + (f_1)^2 \left( \frac{\widehat{\psi}_1}{f_1}-\frac{\widehat{\psi}_2}{f_2}\right)^2 \right] du^\prime \nonumber \\ +  \sup_{(u,v) \in \Delta} \int_{v_0}^v\rho^{-2} \left[  \left(f_1 \partial_v \left( \frac{\widehat{\psi}_1}{f_1}-\frac{\widehat{\psi}_2}{f_2}\right)\right)^2 + (f_1)^2 \left( \frac{\widehat{\psi}_1}{f_1}-\frac{\widehat{\psi}_2}{f_2}\right)^2 \right] dv^\prime \nonumber \\
 \leq C_b \cdot \delta \cdot  \left[  d\left( \left({\tilde{r}}_1, \left({\varpi}_N\right)_1,{\psi}_1\right), \left(\tilde{r}_2,\left({\varpi}_N\right)_2,{\psi}_2 \right) \right)  \right]^2 \, ,
\end{align}
which is almost what we need. To relate it to the honest $C^0\underline{H}^1$-energy we observe
\begin{align}
f_1 \partial_v \left( \frac{\widehat{\psi}_1}{f_1}-\frac{\widehat{\psi}_2}{f_2}\right) = f_1 \partial_v \left(\frac{\widehat{\psi}_1-\widehat{\psi}_2}{f_1} + \frac{\widehat{\psi}_2}{f_2} \left(\frac{f_2}{f_1}-1\right)\right) \nonumber \\
= \tilde{\partial}^{(1)}_v (\widehat{\psi}_1-\widehat{\psi}_2) + \left( \tilde{\partial}^{(2)}_v \widehat{\psi}_2\right) \left(1-\frac{f_1}{f_2}\right)+ f_1\frac{\widehat{\psi}_2}{f_2}\partial_v \left(\frac{f_2}{f_1}-1\right) 
\end{align}
To control the last two terms, we recall that the $H^1$-energy of $\widehat{\psi}_2$ was already established to be bounded. Therefore, the first of them picks up smallness through  (\ref{kaux}) while the second is estimated through (\ref{helpi1}) and picks up smallness via the pointwise bound on $\widehat{\psi}$. As a result we obtain
\begin{align}
\int_{v_0}^v dv^\prime \rho^{-2} |  \tilde{\partial}^{(1)}_v (\widehat{\psi}_1-\widehat{\psi}_2)|^2 \leq \int_{v_0}^v dv^\prime \rho^{-2}\left(f_1 \partial_v \left( \frac{\widehat{\psi}_1}{f_1}-\frac{\widehat{\psi}_2}{f_2}\right)\right)^2 \nonumber \\
+ C_b \cdot \delta^{\min\left(1,2-2\kappa\right)} \cdot  \left[ d\left( \left({\tilde{r}}_1, \left({\varpi}_N\right)_1,{\psi}_1\right), \left(\tilde{r}_2,\left({\varpi}_N\right)_2,{\psi}_2 \right) \right) \right]^2 \, .
\end{align}
Similarly one shows
\begin{align}
\int_{v_0}^v dv^\prime \rho^{-2} | \widehat{\psi}_1 - \widehat{\psi}_2|^2 \leq \int_{v_0}^v dv^\prime \rho^{-2} \Big| \widehat{\psi}_1 - \widehat{\psi}_2 \left(\frac{f_1}{f_2}\right)^2\Big|^2 + \int_{v_0}^v dv^\prime \rho^{-2} (\widehat{\psi}_2)^2 \left(1-\frac{f_1}{f_2}\right)^2 \nonumber \\
\leq \int_{v_0}^v dv^\prime \rho^{-2} \Big| \widehat{\psi}_1 - \widehat{\psi}_2 \left(\frac{f_1}{f_2}\right)^2\Big|^2 + C_b \cdot \delta^{\min\left(1,2-2\kappa\right)}   \left[  d\left( \left({\tilde{r}}_1, \left({\varpi}_N\right)_1,{\psi}_1\right), \left(\tilde{r}_2,\left({\varpi}_N\right)_2,{\psi}_2 \right) \right)  \right]^2 \, , \nonumber
\end{align}
the last step following from the $L^2$-energy of $\widehat{\psi}_2$ being small and (\ref{kaux}). Thus we have established that the left hand side of (\ref{basep1}) actually controls the energy of the difference twisted with $f_1$. Since Lemma \ref{lem:twisteq} establishes the equivalence of twisting with $\rho$ and $\tilde{r}$, we have our desired contraction property
\begin{align} 
 \|\widehat{\psi}_1 - \widehat{\psi}_2\|_{C^0 \underline{H}^1} \leq C_b \cdot \delta^{\min\left(1,2-2\kappa\right)}  \cdot    d\left( \left({\tilde{r}}_1, \left({\varpi}_N\right)_1,{\psi}_1\right), \left(\tilde{r}_2,\left({\varpi}_N\right)_2,{\psi}_2 \right) \right)  \, .\nonumber
\end{align}
The pointwise bounds for the differences $|\rho^{-\frac{3}{2}+\kappa}\left( \widehat{\psi}_1 - \widehat{\psi}_2\right)$ and  $| \rho^{-\frac{1}{2} + \frac{s}{4}}\partial_u \left(\widehat{\psi}_1 - \widehat{\psi}_2\right)|$ then follow as for the estimates for $\widehat{\psi}$ alone (note that the quantities vanish on $v=u_0$) establishing (\ref{goaly2}).

\subsubsection*{The renormalized mass}
The goal is to establish
\begin{align} \label{goaly3}
 d_{\varpi} \left((\varpi_N)_1,(\varpi_N)_2\right) \leq C_b \cdot \delta \cdot    d\left( \left({\tilde{r}}_1, \left({\varpi}_N\right)_1,{\psi}_1\right), \left(\tilde{r}_2,\left({\varpi}_N\right)_2,{\psi}_2 \right) \right) 
\end{align}
This is again a lengthy but straightforward computation. We focus on the most difficult term, which is clearly the first:
\begin{align} \label{iop}
A = \frac{-(r_1)_u}{(\Omega_1)^2} \left(r_1 f_1 \partial_v \left(\frac{\widehat{\psi}_1}{f_1}\right)\right)^2 - \frac{-(r_2)_u}{(\Omega_2)^2} \left(r_2 f_2 \partial_v \left(\frac{\widehat{\psi}_2}{f_2}\right)\right)^2 \nonumber \\
\leq 
 \left(r_1 f_1 \partial_v \left(\frac{\widehat{\psi}_1}{f_1}\right)\right)^2 \left[\frac{(r_2)_u}{(\Omega_2)^2}-\frac{(r_1)_u}{(\Omega_1)^2} \right] \nonumber \\
 + \frac{(r_2)_u}{(\Omega_2)^2} \left[r_1 f_1 \partial_v \left(\frac{\widehat{\psi}_1}{f_1}\right) + r_2 f_2 \partial_v \left(\frac{\widehat{\psi}_2}{f_2}\right) \right] \left[r_1 f_1 \partial_v \left(\frac{\widehat{\psi}_1}{f_1}\right) - r_2 f_2 \partial_v \left(\frac{\widehat{\psi}_2}{f_2}\right) \right] \, ,
\end{align}
where we have used Lemma \ref{lem:cribo}.
Since the twisted $H^1$-energy of both $\hat{\psi}_1$ and $\hat{\psi}_2$ was already shown to be $\delta$-small, one obtains
\[
\Big| \int_{v_0}^v A \left(u,v^\prime\right) dv^\prime \Big| \leq C_b \cdot  \delta \cdot  d\left( \left({\tilde{r}}_1, \left({\varpi}_N\right)_1,{\psi}_1\right), \left(\tilde{r}_2,\left({\varpi}_N\right)_2,{\psi}_2 \right) \right)  \, .
\]
The other terms are handled similarly establishing (\ref{goaly3}) for $\|\left(\widehat{\varpi}_N\right)_1-\left(\widehat{\varpi}_N\right)_2\|_{C^0}$ on the left hand side. To get the statement for $\|\rho^{1+s} \partial_u \left(\widehat{\varpi}_N\right)_1 - \rho^{1+s} \partial_u \left(\widehat{\varpi}_N\right)_2\|_{C^0}$ one differentiates the expression in the contraction map. We again concentrate on the most difficult term (as all other terms simply lose one power of $\tilde{r}$ which is overcome by multiplying with the $\rho^{1+s}$-weight in the end). The most difficult term in $\partial_u A$ is the one involving $u$ derivatives of $\widehat{\psi}$ as for this we only have the pointwise bound (losing $\rho^{-s/2}$) available. Finally,
\begin{align}
\int_{v_0}^v dv^\prime |\partial_u A \left(u,v^\prime\right)| \leq C_b \cdot \rho^{-1-\frac{s}{2}} \left(u,v\right) \cdot  d\left( \left({\tilde{r}}_1, \left({\varpi}_N\right)_1,{\psi}_1\right), \left(\tilde{r}_2,\left({\varpi}_N\right)_2,{\psi}_2 \right) \right)  \, ,
\end{align}
and smallness is obtained after multiplying by $\rho^{1+s}$.\footnote{Actually, most terms already have a $\delta$-smallness in them. However, the term arising from differentiating the last term in (\ref{iop})
\[
 \frac{(r_2)_u}{(\Omega_2)^2} \left[-\frac{\partial_v \left(f_1 r_1\right)}{f_1 r_1} r_1 \tilde{\partial}^{(1)}_u \widehat{\psi}_1 +\frac{\partial_v \left(f_2 r_2\right)}{f_2 r_2} r_2 \tilde{\partial}^{(2)}_u \widehat{\psi}_2\right] \left[r_1 f_1 \partial_v \left(\frac{\widehat{\psi}_1}{f_1}\right) - r_2 f_2 \partial_v \left(\frac{\widehat{\psi}_2}{f_2}\right) \right]
\]
does not.
}

\subsection{Propagating the constraints} \label{sec:constraintprop}

Now, thus far we have established that there exists a solution of  \eq{e1}, \eq{e2}, \eq{e3} with $(\rt, \varpi_N, \psi) \in \mathcal{B}_b$. This alone is insufficient to enable us to reconstruct a solution of Einstein's equations. We need to also establish that the constraint equation, \eq{e4} is satisfied in the region $\Delta$. We proceed by showing that we can propagate the constraint through $\Delta$ using a transport equation in the $v$-direction. 

We first wish to establish that the transport equation for $\varpi_N$ may be differentiated in $u$. We first rewrite \eq{e3}, simplifying the twisted derivatives and making use of the expression for $\Omega$  in terms of $\varpi_N, \psi, r$ to get:
\bea
\partial_v \varpi_N &=& -\frac{4 \pi}{r_v} \left( r \df_v \psi  + g r_v \psi\right)^2 \left( \frac{\varpi_N}{r}\right) + \frac{2 \pi}{r_v} \left( r \df_v \psi + g r_v \psi\right)^2 \nonumber \\&&\quad + \frac{2 \pi}{r_v} \frac{r^4}{l^2}\left(1-4\pi g \psi^2 \right) \left( \df_v \psi\right)^2\label{varpineqn} \\
&& \quad -16 \pi g^2 \frac{r^3}{l^2} \psi^3 \df_v \psi - 8 \pi g^3 r^2 r_v \psi^4\nonumber
\eea
Now, we claim that the right hand side may be differentiated in $u$, with the resultant expression belonging locally to $C^0_u L^1_{v}$. Since we know that $(\rt, \varpi_N, \psi) \in \mathcal{B}_b$, we have that $r \in C^1(\Delta^\circ_{\delta, u_0})$, $\varpi_N, \partial_u (\varpi_N), r_{uv}, \psi, \df_u \psi \in C^0(\Delta_{\delta, u_0})$, $\df_v \psi \in C^0_u L^2_v(\Delta_{\delta, u_0})$. Finally, we note that the wave equation may be written in the form
\begin{align} \label{weqtrans2}
\partial_u \left( r \df_v \psi \right) = - (1+g) r_v \df_u \psi -\frac{\Omega^2}{4} r V \psi 
\end{align}
whence we deduce that $\partial_u (r \df_v \psi) \in C^0$. Now, on differentiating \eq{varpineqn} with respect to $u$, the only terms which are not manifestly in $C^0$ (and hence $C^0_u L^1_v$) are those involving $\df_v \psi$. These are either of the form $f_1 (\df_v \psi)^2$ or $f_2 \df_v \psi$, where $f_i \in C^0$. The terms quadratic in $\df_v \psi$ are manifestly in $C^0_u L^1_v$ since $\df_v \psi \in C^0_u L^2_v$. The terms linear in $\df_v \psi$ can be dealt with by the Cauchy-Schwarz inequality. We thus conclude from the Lemma of the appendix:
\begin{Lemma} \label{varpiureg}
Suppose $(\rt, \varpi_N, \psi)\in \mathcal{B}_b$ solves \eq{e1}, \eq{e2}, \eq{e3}, with suitable boundary conditions imposed on $\scri$ together with initial data on $v=u_0$ according to Definition \ref{def:id}. Then the weak derivative $\partial_v \partial_u(\varpi_N)=\partial_u \partial_v(\varpi_N)$ exists and belongs locally to $C^0_u L^1_v$. 
\end{Lemma}

Having established that we can differentiate equation \eq{e3} with respect to $u$, we can show
\begin{Corollary}\label{constcor}
Under the assumptions of Lemma \ref{varpiureg}, the constraint equation \eq{e4} holds on $\Delta_{\delta, u_0}$.
\end{Corollary}
\begin{proof}
We rewrite the constraint equation as 
\bean
0=\omega&:=&-\partial_u \varpi_N  -\frac{4 \pi}{r_u} \left( r \df_u \psi  + g r_u \psi\right)^2 \left( \frac{\varpi_N}{r}\right) + \frac{2 \pi}{r_u} \left( r \df_u \psi + g r_u \psi\right)^2 \\&&\quad + \frac{2 \pi}{r_u} \frac{r^4}{l^2}\left(1-4\pi g \psi^2 \right) \left( \df_u \psi\right)^2 \\
&& \quad -16 \pi g^2 \frac{r^3}{l^2} \psi^3 \df_u \psi - 8 \pi g^3 r^2 r_u \psi^4 
\eean
We note that under the assumptions, $\omega\in C^0$. Differentiating \eq{e3} with respect to $u$, which we can do by the previous Lemma, and making use of the equations satisfied by $\rt, \psi, \varpi_N$ it is a matter of (tedious)  calculation to show that $\omega$ satisfies
\be
\partial_v \omega = -\frac{4\pi \rt}{\rt_u}\left(\df_v \psi + g \frac{r_v}{r} \psi\right)^2  \omega.
\ee
Now, since the coefficient on the right hand side is locally in $C^0_uL^1_v$, we conclude immediately that $\omega = 0$ on $\Delta_{\delta, u_0}$ (see again the Lemma in the appendix). 
\end{proof}

\section{Generalizations} \label{sec:generalizations}

\subsection{Removing the restriction $\kappa<\frac{2}{3}$} \label{sec:kap}
In this section we provide a sketch of how to extend our results to the entire range $0<\kappa<1$. Let $\kappa \geq \frac{2}{3}$. The first step is to observe that the only non-integrable term in (\ref{e3}) and (\ref{e4}) (cf.~Section \ref{sec:restrict}) can be rewritten as a boundary term and an integrable term
\begin{align} \label{hamop1}
-8\pi^2 g^3 \frac{r^2}{l^2} r_v \psi^4 = \frac{1}{1-\frac{4}{3}g} \left[ -\partial_v \left(\frac{8}{3}\pi^2 \frac{g^3}{l^2} r^3 \psi^4\right) + \frac{32}{3} \pi^2 \frac{g^3}{l^2} r^3 \psi^3 \tilde{\partial}_v \psi \right]  \ \ \ \ \ \textrm{for $\kappa \neq \frac{3}{4}$}
\end{align}
\begin{align} \label{hamop2}
-8\pi^2 g^3 \frac{r^2}{l^2} r_v \psi^4 = -\partial_v \left(8\pi^2 g^3 \frac{1}{l^2}\psi^4 r^3 \log r \right) + 32\pi^2 g^3 \frac{1}{l^2}\psi^4 r^3 \tilde{\partial}_v \psi \log r \ \ \ \ \ \ \textrm{for $\kappa = \frac{3}{4}$}
\end{align}
which allows to define a ``doubly renormalized" Hawking mass
\begin{align}
\varpi_{ND} &=  \varpi_N + \frac{8}{3} \pi^2\frac{g^3}{1-\frac{4}{3}g}  \frac{r^3}{l^2} \psi^4  \ \ \ \ \ \textrm{for $\kappa \neq \frac{3}{4}$} \, , \\
\varpi_{ND} &=  \varpi_N + 8\pi^2 g^3 \frac{r^3}{l^2}\psi^4 \log r  \ \ \ \ \ \ \textrm{for $\kappa \neq \frac{3}{4}$} \, ,
\end{align}
which is expected to be finite for all $0<\kappa<1$. In the second step one reformulates the entire system (\ref{e3})-(\ref{e1}) as equations for the variables $\left(\tilde{r},\varpi_{ND},\psi\right)$ and sets up the analogue of the contraction map. Here we recall that the only other occasion where a restriction on $\kappa$ entered was in the energy estimate for $\psi$, cf.~(\ref{enres}). To remove that obstruction, we observe that it arose from the term proportional to $\psi^2$ in the potential (\ref{e2b}), which did not decay sufficiently strongly (i.e.~at least like $r^{-2}$) for $\kappa<\frac{2}{3}$. However, one can easily see (formally) that this term will enter as an (infinite\footnote{\label{ft:are} For $\kappa \geq 3/4$ the divergent boundary term needs to be removed by a further renormalization of $\psi$, as done for the Hawking mass in (\ref{hamop1}) and (\ref{hamop2}).} for $\kappa\geq \frac{3}{4}$) \emph{boundary term} in the energy estimate for the wave equation.  Therefore, assuming a well-posedness theorem for the non-linear equation $\Box_g \psi = \psi^3$ on asymptotically AdS spacetimes (which is strongly suggested by the a-priori energy estimate that can be derived for this equation) we can replace (\ref{oldpsi}) in the contraction map by defining $\widehat{\psi}$ as the unique solution of 
\begin{align}
-\tilde{\nabla}^\dagger_\mu \tilde{\nabla}^\mu \widehat{\psi} - \frac{\widehat{\psi}^3}{l^2}\left(\kappa-\frac{3}{2}\right) \left[8\pi a +4\pi \left(\kappa-\frac{3}{2}\right)^2 \right]  - V_{reg}\left(\psi,\varpi, r\right) \psi = 0 
\end{align}
with boundary condition $
\rho^{-\frac{1}{2}-\kappa} \left(\tilde{\partial}_{u} - \tilde{\partial}_v\right) \hat{\psi} + 2\beta \left(t\right) \rho^{-\frac{3}{2}+\kappa} \hat{\psi}=\gamma\left(t\right)$. Note that the $\hat{\psi}^3$ term has the correct sign to appear as a positive term in the energy estimate (in case that $\kappa \geq 3/4$ that term is however divergent and a further renormalization is needed, cf.~footnote \ref{ft:are}). Now the potential $V_{reg}$ (which is the potential $V$ of (\ref{e2b}) minus its ``divergent" part) is regular for all $0<\kappa<1$ and the contraction property is established as before using the non-linear well-posedness theory for $\Box_g \psi = \psi^3$ on a fixed aAdS background. It may be that $\Box_g \psi = \psi^3$ is well posed only at a higher level of regularity, in which case one should work at the $H^2$-level, as in \cite{Holzegel:2011qk}.

\subsection{Nonlinear potentials} \label{sec:nonlinear}
Examining the proof of the main theorem, we see that the only properties of the function $V$ we use are an $L^2$-boundedness condition to ensure that we map into the ball, together with a Lipschitz condition to ensure the map contracts. Thus we can readily verify that the above theorem generalizes to non-linear scalar fields with energy momentum tensor
\[
T_{\mu \nu} = \partial_\mu \psi \partial_\nu \psi - \frac{1}{2} g_{\mu \nu} \left[ \left(\partial \psi \right)^2 + \frac{2a}{l^2}\psi^2+W_0\left(\psi\right)\right] 
\]
provided $W_0\left(\psi\right)$ satisfies:
\begin{enumerate}[i)]
\item For any $(\rt, \varpi_N, \psi) \in \mathcal{B}_b$  we have
\be
\int_\Delta \left[\Omega^2 r W_0'(\psi)\right]^2 \leq C_b \delta^\varepsilon
\ee
for some $\varepsilon>0$.
\item For any $(\rt_i, (\varpi_N)_i, \psi_i) \in \mathcal{B}_b$ we have
\be
\int_\Delta \left[\Omega_1^2 r_1 W_0'(\psi_1)-\Omega_2^2 r_2 W_0'(\psi_2)\right]^2 \leq C_b \delta^{\varepsilon'} d((\rt_1, (\varpi_N)_1), \psi_1), (\rt_2, (\varpi_N)_2), \psi_2))
\ee
for some $\varepsilon'>0$. 
\end{enumerate}

 This is of interest in applications, see for example \cite{Gonzalez:arXiv1309.2161} where a potential corresponding to 
 \be
 W_0(\psi) = -\frac{1}{l^2} \left(\cosh \sqrt{2} \psi -1 - \psi^2 \right)+ K \left[6 \sinh \sqrt{2}\psi - 2 \sqrt{2} \psi (2+\cosh\sqrt{2} \psi)\right],
 \ee
 is considered. This potential satisfies $i), ii)$ above provided that $\kappa<\frac{2}{3}$, however, as for the minimally coupled case, we expect this is merely a technical restriction and that the result could be improved to the whole range (in \cite{Gonzalez:arXiv1309.2161}, $\kappa = \frac{1}{2}$). For related work where the metric is assumed to have hyperbolic rather than spherical symmetry, see \cite{Martinez:2004nb, Kolyvaris:2009pc}. This potential with $K=0, \kappa=\frac{1}{2}$ is also considered in \cite{Hertog:2004dr}, where it comes from $\mathcal{N} = 8, D = 4$ gauged supergravity (the massless sector of the compacti�cation of $D = 11$ supergravity on $S^7$) after truncation to an abelian $U(1)^4$ sector.
 
 We note that including scalar fields with several components should also represent a straightforward generalisation of our proof.

\subsection{Non-linear boundary conditions} \label{sec:nlbc}
With a minor modification of the proof, the above theorem also generalizes to certain non-linear (and in principle non-local) boundary conditions. In particular we can consider boundary conditions of the form
\begin{align} \label{bcpsinl}
\rho^{-\frac{1}{2}-\kappa} \left(\tilde{\partial}_u - \tilde{\partial}_v \right) \psi + 2\beta\left(t\right) \rho^{-\frac{3}{2} + \kappa} \psi = G\left[ \rho^{-\frac{3}{2} + \kappa} \psi\right]  \ \ \  \textrm{on $\scri$}
\end{align}
where $G:  H^{\kappa}(\scri) \to H^{1-\kappa}(\scri)$ satisfies the Lipschitz condition
\be
\norm{G[p_1]-G[p_2]}{H^{1-\kappa}(\scri)} \leq K \norm{p_1-p_2}{H^\kappa(\scri)}
\ee
with $K<1$ for all $p_i \in H^\kappa(\scri)$. In the case $\kappa > \frac{1}{2}$, if we take $G[p] = F(p, t)$ with $F:\R\times \scri \to \R$ assumed to be $C^1_{loc.}$, we may arrange that this condition is satisfied by taking $\delta$ sufficiently small. 

In the case $\kappa<\frac{1}{2}$, we appear to only be able to establish well posedness for non-linear boundary conditions which are also non-local. The reason for this is that to have a solution in the energy class for the linear wave equation with inhomogeneous Robin conditions we require (in the absence of further structure) that the inhomogeneity be at least $H^{1-\kappa}$, however the trace theorem only guarantees a trace in $H^\kappa$.

Nonlinear boundary conditions have been considered, for example in \cite{Hertog:2004dr,Hertog:2004ns}. The conditions considered in these papers are of the form (in 3+1 dimensions)
\begin{align} \label{bcpsinl}
\rho^{-\frac{1}{2}-\kappa} \left(\tilde{\partial}_u - \tilde{\partial}_v \right) \psi =k \left[ \rho^{-\frac{3}{2} + \kappa} \psi\right]^{\frac{3 + 2 \kappa}{3 - 2 \kappa}}  \ \ \  \textrm{on $\scri$}
\end{align}
Our results extend to these boundary conditions for $\kappa>\frac{1}{2}$, as well as the case $\kappa=\frac{1}{2}$ provided a smallness assumption is made on the data at infinity.

\section{Improving the regularity: Proof of Theorem \ref{highregthm} }\label{sec:improreg}

Having established that we can always construct a unique weak solution to the renormalized \EKG system of equations subject to appropriate boundary conditions, we wish now to demonstrate that higher regularity is in fact propagated by the equations. Our approach to this will be to show that the contraction map we constructed in \S \ref{contractionsec} in fact respects a certain subspace of $\mathcal{B}_b$ which consist of functions with better regularity than a generic element of $\mathcal{B}_b$. In essence we establish that by commuting the contraction map with the vector field\footnote{It will be convenient also to define $\tilde{T}\psi = \rt^{-\frac{3}{2}+\kappa}T(\rt^{\frac{3}{2}-\kappa} \psi)$ as the twisted $T$-derivative of $\psi$.} $T = \partial_u + \partial_v$ we preserve much of the structure. As a result, the subspace of elements of $\mathcal{B}_b$ whose $T-$derivative also belongs to a ball in the metric space $\mathcal{C}$ is preserved by the contraction map. We first give more stringent conditions on the initial data which guarantee that they represent the restriction to the initial data of a more regular solution to our equations.

\subsection{Constructing higher regularity initial data}
\subsubsection*{Motivation}
In order to construct solutions with higher regularity, we will of course need to assume better regularity for the initial data. Before we do so, we motivate the assumptions we make by recalling some facts from \cite{Warnick:2012fi}. For a solution $\psi$ of the Klein-Gordon equation on a fixed asymptotically AdS background, at the $H^2$ level one finds that $\psi$ should have an expansion near $\scri$ of the form:
\be
\psi=  \rho^{\frac{3}{2}-\kappa}\psi^- +  \rho^{\frac{3}{2}+\kappa} \psi^+  + \O{\rho^{\frac{5}{2}}}
\ee
where $\psi^\pm$ are functions on $\scri$, and we have $\psi^- \in H^1(\scri)$, $\psi^+ \in L^2(\scri)$. Moreover, we have
\be
T\psi =  \rho^{\frac{3}{2}-\kappa}T\psi^-+ \O{\rho^{\frac{3}{2}}}
\ee
for any vector field $T$, tangent to $\scri$ and
\be
\rho^{\frac{3}{2} - \kappa} \partial_\rho \left( \frac{\psi}{\rho^{\frac{3}{2} - \kappa}}\right) = 2\kappa \rho^{\frac{1}{2}+\kappa} \psi^+  + \O{\rho^{\frac{3}{2}}}
\ee
for the twisted derivative normal to the boundary. As a result, we expect that the null derivatives of an $H^2$ solution to the Klein-Gordon equation on an asymptotically AdS background should have expansions:
\bean
2\rho^{\frac{3}{2} - \kappa} \partial_u \left( \frac{\psi}{\rho^{\frac{3}{2} - \kappa}}\right)  &=&  \rho^{\frac{3}{2}-\kappa}\dot{\uppsi} -  \rho^{\frac{1}{2}+\kappa} \uppsi' + \O{\rho^{\frac{3}{2}}} \\
2\rho^{\frac{3}{2} - \kappa} \partial_v \left( \frac{\psi}{\rho^{\frac{3}{2} - \kappa}}\right)  &=&  \rho^{\frac{3}{2}-\kappa}\dot{\uppsi} +  \rho^{\frac{1}{2}+\kappa} \uppsi' + \O{\rho^{\frac{3}{2}}}
\eean
for some functions $\dot{\uppsi}, \uppsi' \in L^2(\scri)$. Note that, as expected, the null derivatives decay like $\rho^{\min(\frac{3}{2} - \kappa, \frac{1}{2}+\kappa)}$. Restricting to an initial data surface we have some necessary conditions on the asymptotic behaviour of initial data which develops into an $H^2$ solution. Of course, the spacetimes that we construct are not asymptotically AdS in as strong a sense as those studied in \cite{Warnick:2012fi}. This manifests itself in part in the subtle distinction between twisting with respect to $\rho$ and $\rt$, and accordingly also in the asymptotic expansions.

\subsubsection{Constructing the data} We now give conditions on a free data set, $(\overline{\rt}, \overline{\psi})$ (with associated full data set $(\overline{\rt}, \overline{\psi}, \overline{\varpi_N}, \overline{\rt_v})$)  such that we can construct the functions  $(\overline{T\rt}, \overline{\tilde{T}\psi}, \overline{T\varpi_N})$ which generate a jet on $\mathcal{M} = \{ (u, v) \in \Delta_{u_0, \delta}: v = u_0\}$ satisfying the equation and boundary conditions there. We first note that we already have constructed $\overline{T\rt} = \overline{\rt}_u + \overline{\rt_v}$. 

In order to construct $\overline{\tilde{T}\psi}$, we will impose some conditions on the behaviour of $\tilde{\partial}_u\overline{\psi}$ near $u_0$. As discussed above, these conditions are necessary in order that the data launch an $H^2$ solution of the Klein-Gordon equation. In particular we require:
\begin{itemize}
\item Defining $\Psi' := \gamma(t_0) - 2 \beta(t_0) \lim_{u\to u_0} \overline{\rho}^{-\frac{3}{2}+\kappa} \overline{\psi}$, we have: 
\be
\overline{\rho}^{\kappa-\frac{3}{2}}\left[ \overline{f}\partial_u \left( \frac{\overline{\psi}}{\overline{f}}\right) - \frac{1}{2}\overline{\rt}^{\frac{1}{2}+\kappa}\Psi'\right] \in C^0(\overline{\mathcal{N}})
\ee
\item Defining 
\be
\dot{\Psi} := 2 \lim_{u \to u_0} \overline{\rho}^{\kappa-\frac{3}{2}}\left[ \overline{f}\partial_u \left( \frac{\overline{\psi}}{\overline{f}}\right)- \frac{1}{2}\overline{\rt}^{\frac{1}{2}+\kappa}\Psi'\right] 
\ee
we furthermore require
\be
\psi_R (u) := \overline{\rho}^{-\frac{3}{2}}\left[ \overline{f}\partial_u \left( \frac{\overline{\psi}}{\overline{f}}\right) - \frac{1}{2} \overline{\rt}^{\frac{1}{2}+\kappa}\Psi'-\frac{1}{2}\overline{\rt}^{\frac{3}{2}-\kappa} \dot{\Psi} \right] \in C^0(\overline{\mathcal{N}})
\ee 
so that:
\be
\tilde{\partial}_u \overline{\psi} : = \frac{1}{2}\overline{\rt}^{\frac{3}{2}-\kappa} \dot{\Psi}+\frac{1}{2} \overline{\rt}^{\frac{1}{2}+\kappa}\Psi'+ \overline{\rt}^{\frac{3}{2}} {\psi}_R(u)
\ee
\item If these conditions on $\overline{\psi}$ hold, then we can construct the function
\be
\overline{\tilde{\partial}_v \psi} : = \frac{1}{2}\overline{\rt}^{\frac{3}{2}-\kappa} \dot{\Psi}-\frac{1}{2} \overline{\rt}^{\frac{1}{2}+\kappa}\Psi'+ \overline{\rt}^{\frac{3}{2}} \tilde{\psi}_R(u)
\ee
where $\tilde{\psi}_R(u) \in  C^0(\overline{\mathcal{N}})$ is defined by
\be
\tilde{\psi}_R(u):= -\overline{\rt}^{-\frac{1}{2}} \int_{u_0}^u \left. \frac{T(\overline{\rt}) \tilde{\partial}_u \overline{\psi}}{\overline{\rt}^2} + \frac{1}{4 \overline{\rt}} \overline{\Omega}^2 \overline{V}\overline{\psi}- \left (\frac{1}{2}-\kappa\right) \frac{\psi_R}{\overline{\rt}^{\frac{1}{2}}} \right|_{u'} du'
\ee
with
\be
\overline{V} := 2g^2 \overline{\rt}^3 \left(\overline{\varpi_N} + 2\pi g\frac{\overline{\psi}^2}{l^2 \overline{\rt}^2} \right)- 8\pi g \frac{a}{l^2} \overline{\psi}^2 - \overline{\rt}^2 \left(\kappa^2 - 2\kappa + \frac{3}{4}\right) \, .
\ee
and
\be
\overline{\Omega}^2 := -4\overline{r}_u \overline{r_v} \left (1-\frac{2\overline{\varpi_N}}{\overline{r}} + \frac{\overline{r}^2}{l^2} -4\pi g \frac{\overline{r}^2}{l^2}\overline{\psi}^2 \right)^{-1} \, .
\ee
This is enough to define $\overline{\tilde{T}(\psi)} =\tilde{\partial}_u \overline{\psi} + \overline{\tilde{\partial}_v \psi} $,  with $ \overline{\rho}^{-\frac{3}{2}+\kappa}\overline{\tilde{T}(\psi)} \in C^0(\overline{\mathcal{N}})$ and to verify that the boundary condition
\be
\overline{\rho}^{-\frac{1}{2}-\kappa} \left(\tilde{\partial}_u \overline{\psi} - \overline{\tilde{\partial}_v \psi} \right) + 2 \beta(t_0) \overline{\rho}^{-\frac{3}{2}+\kappa} \overline{\psi} \to \gamma(t_0), \qquad \textrm{ as } u \to u_0
\ee
holds. Moreover, the Klein-Gordon equation
\be
\partial_u \left( \overline{r} \overline{\tilde{\partial}_v \psi}\right) = -\left( \frac{1}{2}-\kappa\right)\overline{\partial_v r} \overline{\rt}^{\frac{3}{2}-\kappa} \left(\partial_u \frac{\overline{\psi}}{\overline{f}}\right) - \frac{\overline{\Omega}^2}{4} \overline{r} \overline{V}\overline{\psi}, \nonumber
\ee
holds on $\mathcal{N}$. We are going to assume that $\overline{\tilde{T}\psi}$ satisfies the conditions imposed on $\overline{\psi}$ in \S \ref{initdatasec}, which in particular will imply that $\overline{\psi} \in C^2_{loc.}$.
\item We are now in a position to define
\ben{Tvarpiinitial}
\begin{split}
\overline{\partial_v \varpi_N}    :=& -8\pi \overline{r}^2 \frac{\overline{{r}}_u}{{\overline{\Omega}}^2}  \overline{\tilde{\partial}_v \psi}^2+ 4\pi g \left(\overline{r}-2\overline{\varpi_N}\right) \overline{\psi}  \overline{\tilde{\partial}_v \psi} \\
& +  2\pi \overline{\psi}^2  \overline{r_v} \left( g^2 \left(1-\frac{2\overline{\varpi_N}}{\overline{r}}\right) \right) -16\pi^2g^2 \frac{\overline{r}^3}{l^2} \overline{\psi}^3  \overline{\tilde{\partial}_v \psi}- 8\pi^2 g^3 \frac{\overline{r}^2}{l^2} \overline{r_v} \overline{\psi}^4
\end{split}
\een
and we have as a consequence constructed $\overline{T\varpi_N} := \partial_u \overline{\varpi_N} +\overline{\partial_v \varpi_N}$. We can verify directly that $\overline{T\varpi_N} \in C^0(\overline{\mathcal{N}})$, and we denote
\be
 \dot{M}_N:= \lim_{u \to u_0} \overline{T\varpi_N} =  \frac{4 \pi}{l^2} \Psi' \dot{\Psi}
\ee
\item Finally, we may construct $\overline{\rt_{vv}}$ by integrating the linear ODE
\be
\partial_u(\overline{\rt_{vv}}) = \overline{\alpha} \overline{\rt_{vv}} + \overline{\alpha_v}
\ee
with the initial condition\footnote{There is some freedom in how we choose boundary conditions for the higher derivatives of $\rt$ on the initial data, but we choose a convenient gauge in which $T\rt$ vanishes at $\scri$ to all orders on the initial data.} that $\overline{\rt_{vv}}(u_0) = 0$. Here $\overline{\alpha}$ is the restriction to the initial data of the quantity
\be
\alpha:=   \frac{\Omega^2}{{\rt_v} {r}^2} \left(\frac{3 {\varpi_N}}{2 {r}^2} -\frac{1}{2{r}} + \frac{2 \pi  {r} {\psi}^2}{l^2} \left(-a+\frac{3}{2}g\right)\right)
\ee
and $\overline{\alpha_v}$ is obtained by first differentiating $\alpha$ in $v$ and then restricting to the initial data, using the definition of $\Omega$ to see that no term appears which we have not already constructed on $\overline{\mathcal{N}}$. Doing this, we can verify that both $\overline{\alpha}$, $\overline{\alpha_v}$ are integrable in $u$. As a result, we have constructed $\overline{T\rt_v} = \overline{\rt_{vv}} + (\overline{\rt_v})_u$, and we can check that $\overline{T\rt}, (\overline{T\rt})_u, \overline{T\rt_v}$ all vanish at $u=u_0$. We will assume that $\overline{T\tilde{r}}_{uu}\in C^0(\overline{\mathcal{N}})$, and that $\overline{T\tilde{r}}_{uu}(u_0)=0$. This in particular implies that $\overline{\rt}\in C^3_{loc.}$.
\end{itemize}
\begin{Remark}
Note that the Hawking mass at infinity (which requires this level of regularity to define) will not generally be constant in time for the boundary conditions we impose. This is a consequence of the fact that we are permitting energy to enter the space from $\scri$. If we impose homogeneous Neumann boundary conditions, the flux vanishes and the Hawking mass is constant.
\end{Remark}
\begin{Definition}
We say that a free data set $(\overline{\rt}, \overline{\psi})$ gives rise to $H^2-$initial data if it satisfies the conditions given above to allow us to construct $(\overline{T\rt}, \overline{\tilde{T}\psi}, \overline{T\varpi_N})$, and furthermore we have that for any $0<s<1$, the following bounds hold on the initial data ray $\overline{\mathcal{N}}$ for some $C$
\begin{align} \label{Tsr1}
\| \overline{T \tilde{r}} \|_{C^0} + \| (\overline{T \tilde{r}})_u  \|_{C^0}+ \| \overline{T \tilde{r}_v}  \|_{C^0}  + \| \overline{T\tilde{r}}_{uu} \|_{C^0} < C
\end{align}
\begin{align} \label{Tsr2}
\int_{u_0}^{u_0+\delta} \left[ \left(\overline{\tilde{r}}^{-1} \cdot \overline{f} \partial_u  \left( \frac{\overline{\tilde{T} \psi}}{\overline{f}}\right)\right)^2  + \overline{\tilde{r}}^{-2} \overline{\tilde{T} \psi}^2\right] du    < C   \ \ \ \textrm{and} \ \ \ \abs{ \overline{\tilde{r}}^{-\frac{1}{2}+\frac{s}{4}} \cdot \overline{f} \partial_u \left( \frac{\overline{\tilde{T}\psi}}{\overline{f}}\right)}<C
\end{align}
\begin{align} \label{Tsr3}
\| \overline{T \varpi_N}-\dot{M}_N \|_{C^0} < C \ \ \ \textrm{and} \ \ \  \| \overline{\tilde{r}}^{1+s} \partial_u\overline{T \varpi_N} \|_{C^0} < C
\end{align}
\begin{align} \label{Tsr4}
 \| \overline{\tilde{T} \psi} \bar{\rho}^{-\frac{3}{2}+\kappa} - \dot{\Psi} \|_{C^0} < C
\end{align}
For any free data set giving rise to $H^2-$initial data, by truncating the initial data ray we may assume that $C<\delta'$ for any $\delta'>0$.
\end{Definition}

\subsection{The commuted function space}

Recall that in \S \ref{funspacessec} we defined a metric space  $\mathcal{C} = C^{1+}_{\tilde{r}} \left(\Delta_{\delta,u_0}\right)\times C^{0+}_{\varpi_N} \left(\Delta_{\delta,u_0}\right) \times C^{0+}_{\psi} \underline{H}^1\left(\Delta \right)$ with distance
\[
d\left( \left(\tilde{r}_1,(\varpi_N)_1, \psi_1\right) , \left(\tilde{r}_2,(\varpi_N)_2, \psi_2\right)\right) = d_{\tilde{r}} \left(\tilde{r}_1,\tilde{r}_2\right) + d_{\varpi} \left(\left(\varpi_N\right)_1, \left(\varpi_N\right)_2\right) + d_{\psi} \left(\psi_1,\psi_2\right) \, ,
\]
and denoted by $\mathcal{B}_b$ the ball of radius $b$ centred around $\left(\frac{u-v}{2}, M_N,\Psi \rho^{\frac{3}{2}-\kappa}\right)$. We then showed that the map $\Phi:\mathcal{B}_b \to \mathcal{B}_b$ is in fact a contraction map, provided we take the size of the domain $\delta$ to be sufficiently small.

\begin{Definition}
We define the commuted ball $\mathcal{B}^1_b$ to consist of those elements $(\rt, \varpi_N, \psi)$ of $\mathcal{B}_b$ for which we additionally have that $T\rt_v$, $T\rt_{uv}$, $T\rt_{uu}$, $T\varpi_N$, $(T\varpi_N)_u$, $T\psi$, $T\psi_u$ are $C^0$, with the following bounds:
\be
\norm{\frac{TT\rt}{\rho}}{C^0} + \norm{T\rt_{uv}}{C^0}+ \norm{T\rt_{uu}}{C^0} <b\, ,
\ee
\be
\|  T \varpi_N -\dot{M}_N \|_{C^0} + \| \rho^{1+s} \partial_u \left(T \varpi_N\right)_1 \|_{C^0} <b \, ,
\ee
\be
 \| \tilde{T}\psi-\rho^{\frac{3}{2}-\kappa} \dot{\Psi} \|_{C^0\underline{H}^1} +  \| \tilde{T} \psi \rho^{-\frac{3}{2}+\kappa} -\dot{\Psi} \|_{C^0} + \|\rho^{-\frac{1}{2}+\frac{s}{4}} \hat{\partial}_u \tilde{T} \psi \|_{C^0} <b \, ,
\ee
where we define $\tilde{T}\psi = \tilde{\partial}_u \psi +  \tilde{\partial}_v \psi$. We also require that at $u = u_0$ we have
\bean
& \quad \rt(u,u_0) = \overline{\rt}(u),\quad  T\rt(u, u_0) = \overline{T\rt}(u), \quad \varpi_N(u, u_0) = \overline{\varpi_N}(u),& \\ &\psi(u, u_0) = \overline{\psi}(u),\quad \tilde{T} \psi(u, u_0) = \overline{\tilde{T}\psi}(u). &
\eean
\end{Definition}

It is convenient to note the following bounds that can be derived for elements of the commuted ball $\mathcal{B}^1_b$:
\begin{Lemma} \label{Tcommest}
Suppose $(\rt, \varpi_N, \psi) \in \mathcal{B}^1_b$. Then the following estimates hold:
\be
\abs{T \varpi} \leq C_b \cdot \rho^{-2 \kappa}, \qquad \abs{T\left( \frac{\Omega^2}{r^2}\right)} \leq C_b , \qquad |T(V)| \leq C_b \rho^{\min(2, 3-2 \kappa)}
\ee
and
\be
 \abs{T\left( \frac{\rt_u}{\rt}\right)} +  \abs{T\left( \frac{\rt_v}{\rt}\right)} \leq C_b, \quad  \abs{\partial_u\left( \frac{T\rt}{\rt}\right)} +  \abs{\partial_v \left( \frac{T\rt}{\rt}\right)} \leq C_b
\ee
\end{Lemma}
\begin{proof}
The first three estimates follow by direct computation, making use of the fact that we already know $T\rt/\rt$ is bounded. To prove the final estimates, first note, as in Corollary \ref{cor:sc}, we have
\be
\partial_v \left( \rho T\rt_u - \frac{1}{2} T \rt\right) = \rho T\rt_{uv} - \frac{1}{2} TT \rt
\ee
whence we immediately estimate
\be
\abs{ \rho T\rt_u - \frac{1}{2} T \rt} \leq 3 b \int_v^u  (u-v') dv' \leq 3 b \rho^2
\ee
which gives
\be
\abs{\frac{T\rt_u}{\rt} - \frac{T\rt}{2 \rho \rt}} < 3b e^b.
\ee
Now, we note that
\bean
T\left( \frac{\rt_u}{\rt} \right) &=& \frac{T\rt_u}{\rt} - \frac{\rt_u T \rt}{\rt^2} \\
&=& \left(\frac{T\rt_u}{\rt} - \frac{T\rt}{2 \rho \rt}\right) +   \frac{T\rt}{ \rt}\left( \frac{1}{2\rho}  - \frac{\rt_u}{\rt}\right)
\eean
whence it immediately follows that $T\left( \frac{\rt_u}{\rt} \right) $ is bounded by some $C_b$. The $v-$derivative follows in a similar fashion. The final estimate follows by noting that $T\left( \frac{\rt_u}{\rt} \right) = T( \partial_u \log \rt )= \partial_u( T \log \rt )= \partial_u \left( \frac{T\rt}{\rt} \right) $.
\end{proof}

\subsection{Propagation of regularity}
We are now ready to state the main result of this section concerning the propagation of regularity.
\begin{Proposition} \label{propreg}
Suppose that the initial data is in the $H^2$-class. Then the map $\Phi:\mathcal{B}_b \to \mathcal{B}_b$ defined in \S\ref{contractionsec} maps $\mathcal{B}^1_b$ into itself for $\delta$ sufficiently small.
\end{Proposition}

Before we prove this result, we note the following:
\begin{Corollary}\label{highregcor}
Suppose we start with initial data in the $H^2-$class. Then then the weak solution $(\rt, \varpi_N, \psi)\in \mathcal{B}_b$ to the renormalised \EKG system which we constructed above in fact belongs to $\mathcal{B}^1_b$. As a consequence the associated metric $g$ is of class $C^0$.
\end{Corollary}

\begin{proof}[Proof of Proposition \ref{propreg}]
As in \S \ref{contractionsec} we define $(\hat{\rt}, \widehat{\varpi_N}, \hat{\psi}):= \Phi(\rt, \varpi_N, \psi)$. We first note that the conditions
\be
\psi(u, u_0) = \overline{\psi}(u),  \quad \rt(u,u_0) = \overline{\rt}(u), \quad \varpi_N(u, u_0) = \overline{\varpi_N}(u).
\ee
are clearly respected by the contraction map. Now note that the condition that $(\rt, \varpi_N, \psi) \in \mathcal{B}^1_b$ permits us to directly differentiate \eq{req} and establish that
\bean \label{Treq}
T\widehat{\tilde{r}} &=& \overline{r}_u(u)-  \overline{r}_u(v)+ \int_v^udu' \left[\frac{\Omega^2}{r^2} \left(\frac{3 \varpi_N}{2 r^2} -\frac{1}{2r} + \frac{2 \pi  r \psi^2}{l^2} \left(-a+\frac{3}{2}g\right)\right) \right](u', u_0) \\&&+ \int_v^u du^\prime \int_{u_0}^v  dv^\prime T\left[\frac{\Omega^2}{r^2} \left(\frac{3 \varpi_N}{2 r^2} -\frac{1}{2r} + \frac{2 \pi  r \psi^2}{l^2} \left(-a+\frac{3}{2}g\right)\right) \right](u', v')
\eean
We can re-write the first line, using the fact that by \eq{rvinitial}
\be
\overline{r_v}(u)-\overline{r_v}(v)= \int_v^u du' \frac{-4\overline{r}^2\overline{\tilde{r}}_u \overline{\tilde{r}_v}}{1-\frac{2\overline{\varpi_N}}{\overline{r}} + \frac{\overline{r}^2}{l^2} -4\pi g \frac{\overline{r}^2}{l^2}\overline{\psi}^2} \left(\frac{3 \overline{\varpi_N}}{2 \overline{r}^2} -\frac{1}{2\overline{r}} + \frac{2 \pi  \overline{r} \overline{\psi}^2}{l^2} \left(-a+\frac{3}{2}g\right)\right) (u'),
\ee
together with the initial conditions assumed on $(\rt, \varpi_N, \psi)$ to give
\ben{Treq}
T\widehat{\tilde{r}} = \overline{Tr}(u)-  \overline{Tr}(v)+ \int_v^u du^\prime \int_{u_0}^v  dv^\prime T\left[\frac{\Omega^2}{r^2} \left(\frac{3 \varpi_N}{2 r^2} -\frac{1}{2r} + \frac{2 \pi  r \psi^2}{l^2} \left(-a+\frac{3}{2}g\right)\right) \right](u', v')
\een
Clearly we recover from here the condition
\be
T \rt(u,u_0) = \overline{T\rt}(u) .
\ee
Now, since acting on any of the fields with $T$ leaves the behaviour near $u=v$ unchanged, we can repeat the arguments of \S \ref{sec:maptoball} to show that
\be
\norm{\frac{TT\hat \rt}{\rho}}{C^0} + \norm{T\hat \rt_{uv}}{C^0}+ \norm{T\hat \rt_{uu}}{C^0} <b\, ,
\ee
for $\delta$ sufficiently small.

Now let us consider the wave equation. Now note that by the results of \S \ref{sec:well posed} we know that $\hat \psi \in H^2_{loc.}$. As a consequence, since the wave equation holds in $C^0$ along the initial data ray, with $\psi(u, u_0) = \overline{\psi}(u)$, together with the boundary conditions at $(u_0, u_0)$, we can deduce that $\tilde{T} \psi(u, u_0) = \overline{\tilde{T} \psi}(u)$. Moreover, we have sufficient regularity to differentiate the wave equation. Doing so, we deduce that $\widehat{\tilde{T}\psi} := f T(f^{-1} \hat \psi)$ is a weak solution of the wave equation:
\bean
\partial_v \left(f r \left(\partial_u \frac{\widehat{\tilde{T}\psi}}{f}\right)\right) &=& - \partial_u \left(rf\right) \left(\partial_v \frac{\widehat{\tilde{T}\psi}}{f}\right) +F, \ \ \textrm{ or equivalently} \\
\partial_u \left(f r \left(\partial_v \frac{\widehat{\tilde{T}\psi}}{f}\right)\right) &=& - \partial_v \left(rf\right) \left(\partial_u \frac{\widehat{\tilde{T}\psi}}{f}\right) +F, \nonumber
\eean
where
\be
F :=\left(\kappa-\frac{1}{2}\right) \left[ T\left( \frac{\rt_u}{\rt}\right) f r \partial_v \left( \frac{\hat\psi}{f}\right)+T\left( \frac{\rt_v}{\rt}\right) f r \partial_u \left( \frac{\hat\psi}{f}\right)\right] - T\left[ \frac{\Omega^2 V \psi}{f}\right] fr,
\ee
and $\widehat{\tilde{T}\psi}$ weakly satisfies the boundary condition
\[
\rho^{-\frac{1}{2}-\kappa} \left(\tilde{\partial}_{u} - \tilde{\partial}_v\right)\tilde{T} \hat{\psi} + 2\beta \left(t\right) \rho^{-\frac{3}{2}+\kappa} \tilde{T}\hat{\psi}=\gamma'\left(t\right)- 2\beta' \left(t\right) \rho^{-\frac{3}{2}+\kappa} \hat{\psi}.
\]
Now, note that since we control the $C^0\H^1(\Delta)-$norm of $\hat \psi$ in terms of $b$ from the lower order energy estimates, we can immediately bound
\be
\int_\Delta du dv F^2 \leq \frac{b}{100}
\ee
provided that $\delta$ is sufficiently small, making use of the estimates of Lemma \ref{Tcommest}. Proceeding as in \S \ref{sec:maptoball} we deduce that for $\delta$ sufficiently small we have
\be
 \|\widehat{\tilde{T}\psi}-\rho^{\frac{3}{2}-\kappa} \dot{\Psi} \|_{C^0\underline{H}^1} +  \| \widehat{\tilde{T}\psi} \rho^{-\frac{3}{2}+\kappa} -\dot{\Psi} \|_{C^0} + \|\rho^{-\frac{1}{2}+\frac{s}{4}} \hat{\partial}_u \widehat{\tilde{T}\psi} \|_{C^0} <b \, ,
\ee
Finally, note that we actually wish to control $\hat{\tilde{T}}\hat{\psi} = \hat{f} T(\hat{f}^{-1} \psi)$. However, we have that
\be
\hat{\tilde{T}}\hat{\psi}  - \widehat{\tilde{T}\psi} = \left( \frac{T \hat \rt}{\hat \rt} - \frac{T \rt}{\rt} \right)\hat \psi
\ee
and by Lemma \ref{Tcommest} the term in brackets belongs to $C^1(\Delta)$, so for small enough $\delta$ we have 
\be
 \|\hat{\tilde{T}}\hat{\psi}-\rho^{\frac{3}{2}-\kappa} \dot{\Psi} \|_{C^0\underline{H}^1} +  \| \hat{\tilde{T}}\hat{\psi} \rho^{-\frac{3}{2}+\kappa} -\dot{\Psi} \|_{C^0} + \|\rho^{-\frac{1}{2}+\frac{s}{4}} \hat{\partial}_u \hat{\tilde{T}}\hat{\psi} \|_{C^0} <b \, .
\ee
Now let us finally consider the equation for $\varpi_N$. Differentiating in $T$ and making use of the expression \eq{Tvarpiinitial} for $\overline{\partial_v \varpi_N}$ we deduce that
\be
\begin{split} 
T\widehat{\varpi}_N = \overline{T\varpi_N}\left(u\right) + \int_{v_0}^v dv^\prime T\Big[ -8\pi r^2 \frac{{r}_u}{{\Omega}^2}  \left[ f \partial_v \left(\frac{\widehat{\psi}}{f}\right) \right]^2+ 4\pi g \left(r-2\varpi_N\right) \widehat{\psi} \left( f \partial_v \frac{\widehat{\psi}}{f} \right) \\
+  2\pi \widehat{\psi}^2  r_v \left( g^2 \left(1-\frac{2\varpi_N}{r}\right) \right) -16\pi^2g^2 \frac{r^3}{l^2} \widehat{\psi}^3 \left( f \partial_v \frac{\widehat{\psi}}{f} \right) - 8\pi^2 g^3 \frac{r^2}{l^2} r_v \widehat{\psi}^4 \Big] \left(u,v^\prime\right).
\end{split}
\ee
Making use of the bounds for $T-$derivatives on the unhatted functions, together with the bounds derived above for $\hat \psi$, we can again verify that the argument of \S \ref{sec:maptoball} goes through without serious alteration, so that for sufficiently small $\delta$ we have.
\be
\|  T \varpi_N -\dot{M}_N \|_{C^0} + \| \rho^{1+s} \partial_u \left(T \varpi_N\right)_1 \|_{C^0} <b \, ,
\ee
whence we are done.
\end{proof}

\section{Well posedness for the wave equation with rough coefficients} \label{sec:well posed}

In constructing the contraction map in \S \ref{contractionsec} we assumed the following result:
\begin{Proposition}\label{wpprop}
Suppose that $(\rt, \varpi_N, \psi) \in \mathcal{B}_b$ and let $g$ be the metric of the spherically symmetric spacetime defined by these functions, with twisted derivative $\tilde{\nabla}_\mu$. Then there exists a unique solution $\hat{\psi} \in C^0\H^1(\Delta)$ to the wave equation
\ben{wp:weqn}
\tn^\dagger_\mu \tn^\mu  \hat{\psi} = F
\een
with initial conditions
\be
\hat{\psi} = \overline{\psi}\quad  \textrm{ on }\quad  v=u_0
\ee
and boundary conditions
\be
\rho^{-\frac{1}{2}-\kappa} \left(\tilde{\partial}_u - \tilde{\partial}_v \right) \hat{\psi} + 2\beta\left(t\right) \rho^{-\frac{3}{2} + \kappa} \hat{\psi} = \gamma\left(t\right) \ \ \  \textrm{on } \quad \scri
\ee
where $\beta$ is at least $C^1$, provided the spherically symmetric data $F, \overline{\psi}, \gamma$ satisfy
\begin{enumerate}[i)]
\item \be
\int_\Delta F^2 \rho^{-6} du dv<\infty
\ee
\item \be
\int_{u_0}^{u_1} \left[ \left(\overline{f} \partial_u \left( \frac{\overline{\psi}}{\overline{f}}\right)\right)^2  + \overline{\psi}^2\right] \left(u-u_0\right)^{-2} du    <\infty
\ee
\item \be
\norm{\gamma}{H^{1-\kappa}(\scri)} <\infty
\ee
\end{enumerate}
\end{Proposition}
In this section we shall prove this result. Before we do so, let us note that the subtlety here is in the low regularity assumed on the function $\rt$. From the results of \cite{Warnick:2012fi}, the following Lemma follows:
\begin{Lemma}\label{wp:lemma1}
Suppose that in addition to the assumptions above we have that $\rt$ is $C^\infty$ on $\Delta$, and extends smoothly to $\scri$ as an even function\footnote{Equivalently, the extension of $\rt$ across $\scri$ defined by 
\be
\breve{\rt}(u, v)=\left \{ \begin{array}{lcl} 
\rt(u, v) & \quad& v \leq u \\
-\rt(v, u) & \quad& v > u				
\end{array} \right.
\ee 
should be $C^\infty$.} of $\rho$, then Proposition \ref{wpprop} holds.
\end{Lemma}
\begin{proof}
Recall that the wave equation \eq{wp:weqn} takes the form
\be
\partial_v \left(f r \left(\partial_u \frac{\hat {\psi}}{f}\right)\right) = - \partial_u \left(rf\right) \left(\partial_v \frac{\hat{\psi}}{f}\right) - \frac{\Omega^2}{4} r F, 
\ee
As a result, defining $F' = \rt^2 \Omega^2 F$, we see that $\hat{\psi}$ solves \eq{wp:weqn} if and only if it solves
\be
(\tn')^\dagger_\mu (\tn')^\mu  \hat{\psi} = F'
\ee
where $(\tn')^\dagger_\mu (\tn')^\mu$ is constructed from the metric
\be
g' = \frac{-du dv + d\sigma_{S^2}}{\rt^2}.
\ee
This metric satisfies the regularity and boundary conditions of \cite{Warnick:2012fi}, thus the well posedness results of that paper apply, in particular Theorem 6.1. Thus if $i)-iii)$ hold (substituting $F'$ for $F$)  then the conclusions of Proposition \ref{wpprop} hold. Finally, noting that $\rt^2 \Omega^2$ is bounded on $\Delta$, we see that the condition on $F'$ reduces to that on $F$.
\end{proof}

Armed with this result, we are able to prove Proposition \ref{wpprop} by approximating $\rt$, finding the corresponding solutions to the wave equation, and then showing that the sequence of approximations so obtained converges.
\begin{proof}[Proof of Proposition \ref{wpprop}]
Now suppose that $(\rt, \varpi_N, \psi)$ is an arbitrary element of $\mathcal{B}_b$ and fix $F' = \rt^2 \Omega^2 F$. Let $\rt_1, \rt_2$ be two radial functions satisfying the conditions of Lemma \ref{wp:lemma1} and let $\hat{\psi}_i$ be the unique weak solution of
\be
\partial_v \left(f_i r_i \left(\partial_u \frac{\hat {\psi}_i}{f_i}\right)\right) = - \partial_u \left(r_if_i\right) \left(\partial_v \frac{\hat{\psi}_i}{f_i}\right) -\frac{\rt_i^{-3}}{4}  F', 
\ee
such that $\hat{\psi}_i$ satisfies the initial conditions
\be
\hat{\psi}_i = \overline{\psi}\quad  \textrm{ on }\quad  v=u_0
\ee
and boundary conditions
\be
\rho^{-\frac{1}{2}-\kappa} f_i \left(\partial_u - \partial_v \right) \left(\frac{\hat{\psi}_i }{ f_i}\right)+ 2\beta\left(t\right) \rho^{-\frac{3}{2} + \kappa} \hat{\psi}_i = \gamma\left(t\right) \ \ \  \textrm{on } \quad \scri,
\ee
where $f_i = (\rt_i)^{\frac{3}{2}-\kappa}$.

Recall now the estimate from \S \ref{contractionsec} which controls how the solution of the wave equation on the backgrounds defined by two different elements of the ball differ in terms of the distance between the points in the ball. Estimating in precisely the same way\footnote{In \S \ref{contractionsec} we were also able to show that the RHS was $\delta$-small, but this was as a consequence of slightly higher regularity for the inhomogeneity, which we do not assume here.}, but with the simplification now that we may replace terms $\Omega^2_i r_i V_i$ appearing there with $\rt_i^{-3} F'$, we can show that
\begin{align}
 d_\psi(\hat{\psi}_1, \hat{\psi}_2)  \leq C_{b, \overline{\psi}, g, F'} \cdot   d_{\rt}\left({\tilde{r}}_1,\tilde{r}_2 \right) \, .
\end{align}
Clearly, we can take a sequence of points $\rt_i$, with each $\rt_i$ satisfying the postulates of Lemma \ref{wp:lemma1}, such that  $\rt_i \to \rt$ with respect to the $d_{\rt}$ metric. Then the corresponding solutions $\hat{\psi}_i$ of \eq{wp:weqn} converge with respect to the $d_\psi$ metric to the unique weak solution of
\be
\partial_v \left(f r \left(\partial_u \frac{\hat {\psi}}{f}\right)\right) = - \partial_u \left(rf\right) \left(\partial_v \frac{\hat{\psi}}{f}\right) - \frac{\rt^{-3}}{4}  F', 
\ee
which is nothing other than \eq{wp:weqn} on recalling the definition of $F'$.
\end{proof}

We note in passing that a similar result holds for $T\hat{\psi}$ provided $(\rt, \varpi_N, \psi) \in \mathcal{B}^1_b$, with $TF, \overline{T\psi}, T\gamma$ assumed to satisfy $i)-iii)$. This can be deduced by commuting with $T$ and making use of estimates established in \S \ref{sec:improreg}. In particular, we deduce from this that under these assumptions $\hat{\psi} \in H^2_{loc.}$.

\begin{Remark}
We found approximate solutions by first eliminating $\Omega^2$ from the principle part of the operator, and then simply approximating $\rt$. One might instead consider trying to prove this result by considering points $(\rt, \varpi_N, \psi)$ in $\mathcal{B}_b$ such that the metrics they define directly satisfy the regularity conditions of \cite{Warnick:2012fi}. Unfortunately such points are not dense in $\mathcal{B}_b$. To see this, we note that in constructing $\Omega$ from $(\rt, \varpi_N, \psi)$, as in \eq{defaux}, a term of the form $r^2 \psi^2$ appears in the denominator. This will introduce a non-integer power into the expansion of $\Omega$ unless $\psi$ is assumed to vanish sufficiently rapidly near $\scri$, however functions $\psi$ vanishing too rapidly near $\scri$ are not dense in $C^0\H^1$.
\end{Remark}

\newpage

\appendix
\section{The linear equations}
In many places during the course of our arguments, we shall need to estimate properties of solutions to various linear equations in the weakly asymptotically AdS spacetimes. In order to streamline these arguments, we collect in this section some of these estimates.

\subsection{Estimates for the transport equation}

The equation for the renormalised Hawking mass, $\varpi_N$ takes the form of a linear transport equation in $v$, with rough coefficients. We give here a Lemma to allow us to handle such an equation. Firstly let us define the Banach space $C^0_u L^1_v (\overline \Delta_{\delta, u_0})$ which consists of functions on $\delta, u_0$ such that
\be
\norm{\alpha}{C^0_u L^1_v} := \sup_{u \in [u_0, u_0+\delta]} \int_{u_0}^u \abs{\alpha(u, v)} dv < \infty
\ee
We say that\footnote{recall $\Delta_{\delta, u_0} =  \overline{\Delta}_{\delta, u_0} \setminus \scri$} $\phi \in C^0_u L^1_v (\Delta_{\delta, u_0})$ if for any $0<\epsilon<\delta$ we have that
\be
\sup_{u \in (u_0+\epsilon, u_0+\delta]} \int_{u_0}^{u-\epsilon} \abs{\alpha(u, v)} dv < \infty.
\ee
Clearly $C^0_u L^1_v (\overline \Delta_{\delta, u_0})\subset C^0_u L^1_v (\Delta_{\delta, u_0})$.
\begin{Lemma}\label{translem}
Suppose $\alpha, \beta \in  C^0_u L^1_v (\Delta_{\delta, u_0})$ and $\phi_0 \in C^0((u_0, u_0+\delta])$. Then there exists a unique $\phi \in C^0(\Delta_{\delta, u_0})$ such that:
\begin{enumerate}[i)]
\item For each $u \in (u_0, u_0+\delta]$, and $0<\epsilon< u-u_0$, the map $v \mapsto \phi(u, v)$ is absolutely continuous on the interval $[u_0, u-\epsilon]$.
\item The transport equation
\be
\partial_v \phi = \alpha \phi + \beta
\ee
holds for all $u$ and almost every $v$ in $\Delta_{\delta, u_0}$ with the initial condition $\phi(u, u_0) = \phi_0(u)$.
\item  If moreover $\alpha, \beta \in C^0_u L^1_v (\overline \Delta_{\delta, u_0})$ and $\phi_0 \in C^0([u_0, u_0+\delta])$, then: $\phi \in C^0(\overline \Delta_{\delta, u_0})$;  for each $u \in [u_0, u_0+\delta]$ the map $v \mapsto \phi(u, v)$ is absolutely continuous on the interval $[u_0, u]$; and we have the estimate
\ben{odeest}
\norm{\phi}{C^0} \leq e^{2 \norm{\alpha}{C^0_u L^1_v} }\left( \norm{\beta}{C^0_u L^1_v} + \norm{\phi_0}{C^0}\right)
\een
\end{enumerate}
Suppose now that additionally $\partial_u \alpha, \partial_u \beta \in C^0_u L^1_v (\Delta_{\delta, u_0})$ and $\phi_0 \in C^1((u_0, u_0+\delta])$. Then\footnote{We understand the derivative here to be a weak derivative, which will agree with the strong derivative almost everywhere.} $\partial_u\phi \in C^0(\Delta_{\delta, u_0})$ and we have
\begin{enumerate}[i)]
\item For almost every $u \in (u_0, u_0+\delta]$, and for any $0<\epsilon<\delta$, the map $v \mapsto \partial_u \phi(u, v)$ is absolutely continuous on the interval $[u_0, u-\epsilon]$.
\item The equation
\ben{odediff}
\partial_u \partial_v \phi=\partial_v \partial_u \phi = \alpha \partial_u \phi +  ( \partial_u\alpha) \phi+ \partial_u \beta
\een
holds almost everywhere in $\Delta_{u_0, \delta}$.
\item  If moreover $\alpha, \beta, \partial_u \alpha, \partial_u \beta \in C^0_u L^1_v (\overline \Delta_{\delta, u_0})$ and $\phi_0 \in C^1([u_0, u_0+\delta])$ then: $\partial_u \phi \in C^0(\overline \Delta_{\delta, u_0})$;  for each $u \in [u_0, u_0+\delta]$ the map $v \mapsto \partial_u \phi(u, v)$ is absolutely continuous on the interval $[u_0, u]$; and we have the estimate
\ben{odeest2}
\norm{\partial_u \phi}{C^0} \leq 2 e^{2 \norm{\alpha}{C^0_u L^1_v} }\left( \norm{\beta}{C^0_u L^1_v}+ \norm{\partial_u \beta}{C^0_u L^1_v}+ \norm{\partial_u \alpha}{C^0_u L^1_v}  + \norm{\phi_0}{C^1}\right)
\een
\end{enumerate}
\end{Lemma}

\begin{proof}
First we note that if $\alpha \in C^0_u L^1_v(\Delta_{\delta, u_0})$ then the function
\be
\gamma(u, v) := e^{-\int_{u_0}^v \alpha(u, v') dv'}
\ee
belongs to $C^0(\Delta_{\delta, u_0})$, and for any $u_0< u \leq u_0+\delta$, $0<\epsilon<u-u_0$, the map $v \mapsto \gamma(u, v)$ is absolutely continuous on the interval $[u_0, u-\epsilon]$. Furthermore, if $\alpha \in C^0_u L^1_v (\overline \Delta_{\delta, u_0})$  we estimate
\bean 
\abs{\gamma(u, v)} &\leq& e^{\int_{u_0}^v \abs{\alpha(u, v')} dv'} \leq e^{\int_{u_0}^u \abs{\alpha(u, v')} dv'} \\
&\leq& e^{\norm{\alpha}{C^0_u L^1_v} }
\eean
and a similar estimate holds for $\abs{\gamma(u, v)^{-1}}$.

Now let us define
\be
\phi(u, v) = \gamma(u, v)^{-1}\left( \int_{u_0}^v \beta(u, v') \gamma(u, v') dv' + \phi_0(u)\right).
\ee
We readily verify that this is absolutely continuous in $v\in[u_0, u-\epsilon]$, for any $u_0< u \leq u_0+\delta$, $0<\epsilon<u-u_0$. Furthermore, $\phi$ satisfies
\be
\partial_v \phi = \alpha \phi + \beta, \qquad \psi(u, u_0) = \psi_0(u)
\ee
for all $u\in(u_0, u_0+\delta]$ and almost every $v\in[u_0, u)$. To prove uniqueness, suppose $\beta=0$, $\phi_0=0$. We can differentiate $\phi \gamma$ to find
\be
\partial_v (\phi \gamma) = 0
\ee
for all $u$ and almost every $v$, whence $\phi \equiv 0$. Finally we may directly estimate from the equation for $\phi$ to show \eq{odeest} holds if the coefficients are assumed to be globally bounded.

Now we consider the case where  $\partial_u \alpha, \partial_u \beta \in C^0_u L^1_v (\Delta_{\delta, u_0})$ and $\phi_0 \in C^1((u_0, u_0+\delta])$. Since $\alpha$ and $\partial_u \alpha$ are locally integrable on $\Delta_{u_0, \delta}$, we have that
\be
\partial_u \gamma(u, v) = \gamma(u, v) \left(-\int_{u_0}^v \partial_u \alpha(u, v') dv' \right)
\ee
holds almost everywhere. Furthermore, the right hand side is in $C^0(\Delta_{\delta, u_0})$. Directly differentiating the expression for $\phi$ above, making a similar argument to differentiate $\beta$ under the integral sign, we conclude that $\partial_u \phi\in C^0(\Delta_{\delta, u_0})$. Here we must interpret the derivative as a weak derivative, so that continuity holds modulo redefinition on a set of measure zero. Differentiating once more with respect to $v$ we conclude \eq{odediff} holds. Finally, if we make the further assumption that $\alpha, \beta, \partial_u \alpha, \partial_u \beta \in C^0_u L^1_v (\overline \Delta_{\delta, u_0})$ and $\phi_0 \in C^1([u_0, u_0+\delta])$ we can readily estimate \eq{odeest2} by applying the estimate from the previous discussion.
\end{proof}

%
%
%
\bibliographystyle{utphys}
\bibliography{AdS}

\end{document}